\documentclass[11pt]{article}
\usepackage{fullpage}

\usepackage{authblk}
\usepackage[utf8]{inputenc}
\usepackage{listings}
\usepackage{amsmath}
\usepackage{comment}
\usepackage[skins, breakable]{tcolorbox}
\usepackage{multirow}
\usepackage{amsthm}
\usepackage{amsfonts}
\usepackage{graphicx}
\usepackage{appendix}
\usepackage{xspace}
\usepackage{enumitem}
\usepackage{newfloat}
\usepackage[skip=2pt]{caption}
\usepackage{adjustbox}
\usepackage{xpatch}
\usepackage{tcolorbox}
\usepackage{hyperref}

\DeclareCaptionType{InfoBox}

\newtheorem{lemma}{Lemma}
\newtheorem{theorem}{Theorem}
\newtheorem{definition}{Definition}

\newcommand{\fa} {\mathcal{A}}
\newcommand{\fb} {\mathcal{B}}
\newcommand{\fc} {\mathcal{C}}
\newcommand{\fd} {\mathcal{D}}
\newcommand{\fe} {\mathcal{E}}
\newcommand{\ff} {\mathcal{F}}
\newcommand{\fl} {\mathcal{L}}
\newcommand{\fp} {\mathcal{P}}

\newcommand{\fs} {\mathcal{S}}

\newcommand{\security}{\lambda\xspace}
\newcommand{\sig}{k_s\xspace}
\newcommand{\hashsize}{k_h\xspace}
\newcommand{\hash}{\texttt{Hash}\xspace}
\newcommand{\msg}{\texttt{HAPPY}\xspace}

\newcommand{\encode}{\texttt{Encode}\xspace}

\newcommand{\distribute}{\texttt{Distribute}\xspace}

\newcommand{\reconstruct}{\texttt{Reconstruct}\xspace}

\newlist{numlist}{enumerate}{10}
\setlist[numlist]{label*=\arabic*.}

\newcommand{\confset}{\fp_{\texttt{conflict}}\xspace}
\newcommand{\hmsmset}{\fp_{\texttt{hmsm}}\xspace}

\newcommand{\pair}[1]{\ensuremath{\langle #1, s_{#1} \rangle}}

\newcommand{\gen}{\texttt{Gen}\xspace}
\newcommand{\eval}{\texttt{Eval}\xspace}
\newcommand{\createwit}{\texttt{CreateWit}\xspace}
\newcommand{\verify}{\texttt{Verify}\xspace}

\newcommand{\stitle}[1]{\vspace{0.5ex} \noindent\textsf{\textbf{#1}}}

\newtcolorbox{mybox}[1][]{
  enhanced,
  colback=white,
  boxsep=0pt,
  #1} 

\newcommand{\rl}[1]{{\color{red}[Ling: #1]}}

\begin{document}

\title{Improved Extension Protocols for Byzantine Broadcast and Agreement~\thanks{
Elaine Shi is supported in part by NSF award 1561209, and part of the work was done when the author was a long-terms visitor in Simons Institute for the ``Proofs, Consensus, and Decentralizing Societ'' program. 
Nitin H. Vaidya is supported in part by NSF award 1849599.
Kartik Nayak and Ling Ren are supported in part by a VMware early career faculty grant.}}

\author[1]{Kartik Nayak}
\author[2]{Ling Ren}
\author[3]{Elaine Shi}
\author[4]{Nitin H. Vaidya}
\author[5]{Zhuolun Xiang~\thanks{Lead author.}}
\affil[1]{Duke University\\ {kartik@cs.duke.edu}}
\affil[2, 5]{University of Illinois at Urbana-Champaign\\ {\{renling, xiangzl\}@illinois.edu}}
\affil[3]{Cornell University\\ {runting@gmail.com}}
\affil[4]{Georgetown University\\ {nitin.vaidya@georgetown.edu}}

\maketitle

\begin{abstract}
    Byzantine broadcast (BB) and Byzantine agreement (BA) are two most fundamental problems and essential building blocks in distributed computing, and improving their efficiency is of interest to both theoreticians and practitioners. 
    In this paper, we study extension protocols of BB and BA, i.e., protocols that solve BB/BA with long inputs of $l$ bits using  lower costs than $l$ single-bit instances.
    We present new protocols with improved communication complexity in almost all settings: authenticated BA/BB with $t<n/2$, authenticated BB with $t<(1-\epsilon)n$, unauthenticated BA/BB with $t<n/3$, and asynchronous reliable broadcast and BA with $t<n/3$.
    The new protocols are advantageous and significant in several aspects. First, they achieve the best-possible communication complexity of $\Theta(nl)$ for wider ranges of input sizes compared to prior results. 
    Second, the authenticated extension protocols achieve optimal communication complexity given the current best available BB/BA protocols for short messages.
    Third, to the best of our knowledge, our asynchronous and authenticated protocols in the setting are the first extension protocols in that setting. 
\end{abstract}

\renewcommand{\paragraph}[1]{\smallskip\stitle{#1}}

\section{Introduction}
\label{sec:intro}
This paper investigates extension protocols \cite{ganesh2016broadcast} for Byzantine broadcast (BB) and Byzantine agreement (BA).
The goal of BB is for some designated party (sender) to send its message to all parties and let them output the same message, despite some malicious parties that may behave in a Byzantine fashion. 
The goal of BA is to let all parties each with an input message output the same message. 
We are interested in designing efficient BB/BA protocols with \emph{long messages} since such protocols are widely used as building blocks for other distributed systems such as multi-party computation \cite{yao1982protocols} and permissioned blockchain \cite{miller2016honey}.
For example, practical blockchain systems typically achieve agreement on large blocks (e.g., 1MB).

A straightforward solution for BB/BA with $l$-bit long messages is to invoke the single-bit BB/BA oracle $l$ times. 
This approach will incur at least $\Omega(n^2 l)$ communication complexity where $n$ is the number of parties, because any deterministic single-bit BB/BA has cost $\Omega(n^2)$ due to a lower bound in \cite{dolev1982bounds}.
Another tempting solution is to run BB/BA on the hash digest and let parties disseminate the actual message to each other.
However, if a linear fraction of parties can be Byzantine (which is the typical assumption),
they can each ask all honest parties for the long message, again forcing the communication complexity to be $\Omega(n^2 l)$.

It turns out that non-trivial techniques are needed to get better than $\Omega(n^2 l)$ or to achieve the optimal communication complexity of $O(nl)$.
These are known in the literature as \emph{extension protocols}, which construct BB/BA with long input messages using a small number of BB/BA primitives for short messages.
In this paper, we focus on the authenticated setting where cryptographic techniques are used. 
Table \ref{table:summary:crypto} summarizes the most closely related works and our new results on authenticated extension protocols. 
(In Appendix \ref{sec:errorfree}, we will present some improvements to unauthenticated extension protocols.)
In Table~\ref{table:summary:crypto}, $n$ is the number of parties, $t$ is the maximum number of Byzantine parties, $l$ is the length of the input, $\fa(l)$ is the communication cost of $l$-bit BA oracle, and $\fb(l)$ is the communication cost of $l$-bit BB or reliable broadcast oracle. 
Here we describe the related works in the table.
Let $\hashsize$ denote output size of the collision-resistant hash function.
For both Byzantine broadcast and agreement in the synchronous setting under $t<n/2$,
recent work proposes cryptographically secure extension protocols with communication cost $O(nl+n\fb(\hashsize)+n^3\hashsize)$ \cite{ganesh2016broadcast, ganesh2017optimal}. 
For the case of $t<n$, the state-of-the-art cryptographically secure BB extension protocols have communication complexity $O(nl+\fb(n\hashsize)+n^2\fb(n\log n))$ \cite{ganesh2016broadcast,ganesh2017optimal}. 
There exist information-theoretic authenticated protocols~\cite{fitzi2006optimally, chongchitmate2018information} but they have worse communication complexity than cryptographic ones.
To the best of our knowledge, there exist no extension protocols in the authenticated and asynchronous setting when the paper is written~\footnote{A concurrent work \cite{cryptoeprint:2020:842} independently developed an extension protocol for validated Byzantine agreement in the authenticated and asynchronous setting.}.

\begin{table}[t!] \centering
\resizebox{\columnwidth}{!}{
\begin{tabular}{|l|l|l|l|l|l|}
\hline
\textbf{Threshold}                
& \textbf{Model}   
& \textbf{Problem}                      
& \begin{tabular}[c]{@{}l@{}}\textbf{Communication} \\ \textbf{Complexity}\end{tabular}            
& \begin{tabular}[c]{@{}l@{}}\textbf{Input range} $l$ \\ \textbf{to reach optimality}\end{tabular}     
& \textbf{Reference} \\ 
\hline
 $t<n/2$  & sync. & 
\begin{tabular}[c]{@{}l@{}}
agreement/broadcast
\end{tabular}
& 
\begin{tabular}[c]{@{}l@{}} 
$O(nl+n\fb(k)+k n^3)$\\ 
$O(nl+\fa(k)+k n^2)$ \quad\footnotemark
\end{tabular} & 
\begin{tabular}[c]{@{}l@{}} 
$\Omega(n^3+kn^2)$\\ 
$\Omega(n^2+kn)$
\end{tabular} & 
\begin{tabular}[c]{@{}l@{}} 
\cite{ganesh2016broadcast, ganesh2017optimal}\\ 
\textbf{This paper} 
\end{tabular}
 \\
\hline 
$t< n$  & sync. & broadcast
& 
$O(nl+\fb(nk)+n^2\fb(n\log n))$
& 
$\Omega(n^5\log n+kn^4\log n)$
& 
\cite{ganesh2016broadcast, ganesh2017optimal} \\
\hline
$t< (1-\varepsilon)n$  & sync. & broadcast
& 
$O(nl+\fb(k)+kn^2+n^3)$
& 
$\Omega(n^2+kn)$
& 
\textbf{This paper} \\
\hline
$t<n/3$  & async. & 
\begin{tabular}[c]{@{}l@{}} 
agreement \\
reliable broadcast
\end{tabular}   
& 
\begin{tabular}[c]{@{}l@{}} 
$O(nl+\fa(k)+k n^2)$\\
$O(nl+\fb(k)+k n^2)$
\end{tabular}   
&
\begin{tabular}[c]{@{}l@{}} 
$\Omega(kn)$\\
$\Omega(kn)$
\end{tabular}   &
\begin{tabular}[c]{@{}l@{}} 
\textbf{This paper}\\ 
\textbf{This paper}
\end{tabular}  
\\ 
\hline
\end{tabular}
}
\caption{Cryptographically Secure Extension Protocols for Byzantine Agreement and Broadcast}
\label{table:summary:crypto}
\vspace{-8pt}
\end{table}
\footnotetext{Our cryptographic BB extension protocol can achieve $O(nl+\fb(k)+\fa(1)+k n^2)$ in Appendix~\ref{sec:cryptobb:half}.}

\paragraph{Contributions.}
Table~\ref{table:summary:crypto} also presents our improved protocols in the respective settings. 
Several cryptographic primitives have been employed in our work and prior works. 
To make the communication costs comparable, we assume that the output length of the involved cryptographic building blocks are on the same order, and are all represented by $k$.
We will justify this decision in Section~\ref{sec:pre}.

All our protocols achieve the optimal communication complexity $O(nl)$ for wider ranges of input sizes (see Table \ref{table:summary:crypto} above for authenticated protocols and Table \ref{table:summary:errorfree} in Appendix \ref{sec:errorfree} for  unauthenticated protocols).
In particular, our synchronous and authenticated protocols achieve $O(nl)$ communication complexity when the input size is at least $l=\Omega(n^2+kn)$.
For comparison, state-of-art protocol in the literature require a factor of $n$ larger input size for the $t<n/2$ case, and a factor of $n^3\log n$ larger input size for $n/2 \leq t < (1-\varepsilon)n$ where $\varepsilon$ is a constant. But a limitation of our protocol is that it cannot achieve $O(nl)$ communication if $\varepsilon=o(1)$. 
As for the round complexity, all our extension protocols only adds $O(1)$ communication rounds, except the one for $t<n/2$ which adds $O(t)$ rounds.
All our protocols only invoke the BB/BA oracle $O(1)$ times.

In addition to reaching optimality under smaller input size, our authenticated extension protocols have the following advantages.

\begin{itemize}[noitemsep,topsep=0pt]
\item The communication complexity of our BA extension protocols is very close to the lower bound $\Omega(nl+\fa(k)+n^2)$.
In addition, under the current best BA primitives for short messages, they achieve best-possible communication complexity.
In order to improve upon our extension protocols, one must invent BA primitives for short messages with cost $o(kn^2)$, which seems challenging as we discuss in Section~\ref{sec:cryptoba:lb}.

\item Our protocols can be easily adapted to the asynchronous setting. 
To the best of our knowledge, these are the first asynchronous authenticated extension protocols.
\footnote{Asynchronous unauthenticated protocol exist and they can be used in the authenticated setting, but the cost would be much higher than our new protocols (refer to Table~\ref{table:summary:errorfree} in Appendix~\ref{sec:errorfree}.)}

\item Their simplicity makes our protocols less error-prone and more appealing for practical adoption.
On this note, in deriving our results, we discover a flaw in the prior best protocol \cite{ganesh2016broadcast,ganesh2017optimal} and we provide a simple fix (Appendix \ref{sec:previous_bug}).
\end{itemize}

\section{Related Work}
\label{sec:related}

\paragraph{Timing and setup assumptions.}
With different security assumptions on the adversary and timing assumptions, Byzantine broadcast and agreement can be solved for different thresholds of the Byzantine parties. 
For the timing assumptions, protocols under both synchrony and asynchrony have been studied.
If a trusted setup like public-key infrastructure (PKI) exists, it is called the authenticated setting;
otherwise, it is the unauthenticated setting.
In the synchronous setting, BB/BA can be solved under $t<n/3$ without authentication \cite{LSP82};
with authentication, BA can be solved under $t<n/2$ and BB can be solved under $t<n$ \cite{LSP82,dolev1983authenticated,pfitzmann1996information}. In the asynchronous setting, BB is impossible; BA (randomized) and reliable broadcast can be solved under $t<n/3$ with or without authentication \cite{bracha1987asynchronous,cachin2001secure}.

\paragraph{Previous extension protocols.} 
Table \ref{table:summary:crypto} summarizes the two most closely related works on authenticated extension protocols.
Here, we mention several other ones.
Cachin and Tessaro \cite{cachin2005asynchronous} adapt Bracha's broadcast \cite{bracha1987asynchronous} to handle $l$-bit long messages with communication cost $O(nl+\hashsize n^2\log n)$.
Their method partially inspired our work;
but their method does not seem to apply to general protocols and hence does not yield an extension protocol. 
Related unauthenticated extension protocols are summarized in Table~\ref{table:summary:errorfree} in Appendix~\ref{sec:errorfree}.
Liang and Vaidya~\cite{liang2011error} propose the first optimal error-free BB and BA with communication complexity $O(nl+(n^2\sqrt{l}+n^4)\fb(1))$ for the synchronous case.
Patra \cite{patra2011error} improves the communication complexity to $O(nl+n^2\fb(1))$ under synchrony and also extended the protocols to asynchrony with increased communication complexity (see Table~\ref{table:summary:errorfree}).


\paragraph{State-of-the-art oracle schemes.} To better interpret the improvements we obtained for extension protocols, we provide a summary of the state-of-the-art broadcast and agreement protocols that can be used as the oracle in our extension protocol.
Since our extension protocols are all deterministic, we focus on {\em deterministic} solutions for the most part of the paper, 
except for asynchronous BA where randomization is necessary.
The best deterministic solution to authenticated BB for $t<n$ is the classic Dolev-Strong~\cite{dolev1983authenticated} protocol.
After applying multi-signatures, the communication complexity to broadcast $k$ bits is $\fb(k)=\Theta((k+\sig)n^2 + n^3)$ where $\sig$ is the signature size. 
The Dolev-Strong protocol can also be modified to solve authenticated BA for the $t<n/2$ case (BA is impossible if $t \geq n/2$).
Using an initial all-to-all round with multi-signature to simulate the sender, the communication complexity remains as $\fa(k)=\Theta((k+\sig)n^2 + n^3)$.
In the unauthenticated setting, only $t<n/3$ Byzantine parties can be tolerated and Berman et al.~\cite{berman1992bit} achieves $\fb(1)=\fa(1)=\Theta(n^2)$ (when $t=\Theta(n)$), matching the lower bound on communication complexity.

In the asynchronous setting, Bracha's reliable broadcast~\cite{bracha1987asynchronous} is deterministic and has communication complexity $\fb(1) = O(n^2)$.
Randomization is necessary for asynchronous BA given the FLP impossibility \cite{fischer1985impossibility}.
State-of-art protocols rely on ``common coins'' to provide shared randomness but are deterministic otherwise.
The most efficient unauthenticated asynchronous BA \cite{mostefaoui2015signature} achieves expected communication complexity $\fa(1) = O(n^2)$ assuming a common coin oracle.
The most efficient authenticated asynchronous BA \cite{abraham2018validated} achieves expected communication complexity $\fa(k)=O((k+\sig)n^2)$ and provides a construction for the common coin oracle. 


\paragraph{Coding schemes in consensus systems.}
Several works have taken advantage of coding schemes in practical fault-tolerant consensus systems. 
HoneyBadgerBFT \cite{miller2016honey} and BEAT \cite{duan2018beat} use the reliable broadcast proposed by Cachin and Tessaro \cite{cachin2005asynchronous} as a component for broadcasting blocks efficiently. 
Recent works also apply erasure coding to crash-tolerant systems like Paxos \cite{mu2014paxos} and Raft \cite{wang2020craft}.

\section{Preliminaries}\label{sec:pre}
We consider $n$ parties $P_1,...,P_n$ connected by a reliable, authenticated all-to-all network, where up to $t$ parties may be corrupted by an adversary $A$ and behave in a Byzantine fashion.
We consider both the synchronous model, where there exists a known upper bound on the communication and computation delay, and the asynchronous model, where such an upper bound does not exist.
We consider a {\em static} adversary which decides the set of corrupted parties at the beginning of the execution. 
We denote parties that are not corrupted by the adversary as honest parties.
Two types of the adversary are considered: a computationally bounded adversary is considered in cryptographically secure protocols and a computationally unbounded adversary is considered in the error-free protocols.
Our cryptographically secure protocols additionally assume a trusted setup for a public key infrastructure (PKI) and cryptographic accumulators (see Section \ref{sec:tools}).
The {\em communication complexity} \cite{yao1979some} of the protocol is measured by the worst-case or expected number of bits transmitted by the honest parties according to the protocol specification over all possible executions under any adversary strategy.
Here, we provide the formal definition of Byzantine broadcast (BB) and Byzantine agreement (BA). 


\begin{definition}[Byzantine Broadcast]
    A protocol for a set of parties $\fp = \{P_1, ..., P_n\}$, where a distinguished party called the sender $P_s\in \fp$ holds an initial $l$-bit input $m$, is a Byzantine broadcast protocol tolerating an adversary $A$, if the following properties hold
    \begin{itemize}[noitemsep,topsep=0pt]
        \item Termination. Every honest party  outputs a message.
        \item Agreement. All the honest parties output the same message.
        \item Validity. If the sender is honest, all honest parties output the message $m$.
    \end{itemize}
    \label{defn:BB}
\end{definition}

\begin{definition}[Byzantine Agreement]
    A protocol for a set of parties $\fp = \{P_1, ..., P_n\}$, where each party $P_i\in \fp$ holds an initial $l$-bit input $m_i$, is a Byzantine agreement protocol tolerating an adversary $A$, if the following properties hold
    \begin{itemize}[noitemsep,topsep=0pt]
        \item Termination. Every honest party  outputs a message.
        \item Agreement. All the honest parties output the same message.
        \item Validity. If every honest party $P_i$ holds the same input message $m$, then all honest parties output the message $m$.
    \end{itemize}
\end{definition}

For cryptographically secure protocols and randomized protocols, the above properties hold except for a negligible probability in the security parameter. For brevity, our theorem statements will not mention this explicitly.

\subsection{Primitives}

\label{sec:tools}

In this section, we define several primitives that will be used in our extension protocols.
Our extension protocols use standard coding and cryptographic schemes from the literature, such as linear error correcting codes, muti-signature schemes and cryptographic accumulators.

\paragraph{Linear error correcting code \cite{reed1960polynomial}.}
\label{sec:rscodes}
We will use standard Reed-Solomon (RS) codes \cite{reed1960polynomial} in our protocols, which is a $(n,b)$ RS code in Galois Field $\mathbb{F}=GF(2^a)$ with $n\leq 2^a-1$.
This code encodes $b$ data symbols from $GF(2^a)$ into codewords of $n$ symbols from $GF(2^a)$, and can decode the codewords to recover the original data.
\begin{itemize}[noitemsep,topsep=0pt]
    \item \texttt{ENC}. Given inputs $m_1,...,m_b$, an encoding function $\texttt{ENC}$ computes $(s_1,...,s_n)=\texttt{ENC}(m_1,...,m_b)$, where $(s_1,...,s_n)$ are codewords of length $n$. By the property of the RS code, knowledge of any $b$ elements of the codeword uniquely determines the input message and the remaining of the codeword. 
    \item \texttt{DEC}. The function $\texttt{DEC}$ computes $(m_1,...,m_b)=\texttt{DEC}(s_1,...,s_n)$, and is capable of tolerating up to $c$ errors and $d$ erasures in codewords $(s_1,...,s_n)$, if and only if $n-b\geq 2c+d$. 
In our protocol, We will invoke $\texttt{DEC}$ with specific values of $c,d$ satisfying the above relation, and $\texttt{DEC}$ will return correct output.
\end{itemize}
In our extension protocols, we will use the above RS codes with $n$ equal the number of all parties, and $b$ equal the number of honest parties, i.e., $b=n-t$.


\paragraph{Multi-signatures \cite{boneh2003aggregate}.}
\label{sec:multisig}
Multi-signature scheme can aggregate $n$ signatures into one signature, therefore reduce the size of signatures.
Given $n$ signatures $\sigma_i=\texttt{Sign}(sk_i, m)$ on the same message $m$ with corresponding public keys $pk_i$ for $1\leq i\leq n$,
a multi-signature scheme can combine the $n$ signatures above into one signature $\Sigma$ where $|\Sigma|=|\sigma_i|$.
The combined signature can be verified by anyone using a verification function $\texttt{Ver}(PK, \Sigma, m, \fl)$, where $\fl$ is the list of signers and $PK$ is the union of $n$ public keys $pk_i$.

\paragraph{Cryptographic accumulators \cite{baric1997collision, derler2015revisiting}.}
\label{sec:accumulator}
We present the definition of {\em cryptographic accumulators}  proposed by Bari\'c and Pfitzmann \cite{baric1997collision}.
Intuitively, the cryptographic accumulator constructs an accumulation value for a set of values and can produce a witness for each value in the set.
Given the accumulation value and a witness, any party can verify if a value is indeed in the set. 
Formally, given a parameter $k$, and a set $\fd$ of $n$ values $d_1,...,d_n$, an accumulator has the following components:
\begin{itemize}[noitemsep,topsep=0pt]
    \item $\gen(1^k, n)$:
    This algorithm takes a parameter $k$ represented in unary form $1^k$ and an accumulation threshold $n$ (an upper bound on the number of values that can be accumulated securely), returns an accumulator key $ak$. This step is run by a trusted dealer, so the accumulator key $ak$ is known to all parties.
    \item $\eval(ak, \fd)$:
    This algorithm takes an accumulator key $ak$ and a set $\fd$ of values to be accumulated, returns an accumulation value $z$ for the value set $\fd$.
    \item $\createwit(ak, z, d_i)$:
    This algorithm takes an accumulator key $ak$, an accumulation value $z$ for $\fd$ and a value $d_i$, returns $\bot$ if $d_i\notin \fd$, and a witness $w_i$ if $d_i\in \fd$.
    \item $\verify(ak, z, w_i, d_i)$: 
    This algorithm takes an accumulator key $ak$, an accumulation value $z$ for $\fd$, a witness $w_i$ and a value $d_i$, returns {\em true} if $w_i$ is the witness for $d_i\in \fd$, and {\em false} otherwise.
\end{itemize}

For simplicity, our definition of the cryptographic accumulator above omits the auxiliary information $aux$ that appears in the standard definition \cite{baric1997collision} because the bilinear accumulator we will use does not use $aux$.
We also assume that the function \eval is {\em deterministic}, which is the case with the bilinear accumulator.
We give the detailed description of the {\em bilinear accumulator}~\cite{nguyen2005accumulators, kate2010constant} in Appendix \ref{app:bilinear}. 
The bilinear accumulator satisfies the following property.

\begin{lemma}[Collision-free accumulator \cite{nguyen2005accumulators}]
\label{lem:accumulator}
    The bilinear accumulator is collision-free. That is, for any set size $n$ and any probabilistic polynomial-time adversary $\mathcal{A}$, 
    there exists a negligible
    function ${\sf negl}(\cdot)$, such that
    \[
    \Pr\left[
    \begin{array}{l}
    ak \leftarrow \gen(1^k, n),
    (\{d_1,...,d_n\}, d', w') \leftarrow \mathcal{A}(1^k, n, ak), 
    z \leftarrow \eval(ak, \{d_1,...,d_n\}):\\
    (d' \notin \{d_1,...,d_n\})
    \wedge (\verify(ak, z, w', d')=true)
    \end{array}
    \right] \leq {\sf negl}(k)
    \]
\end{lemma}

To better understand the definition of the cryptographic accumulator, it is helpful to note that the Merkle tree \cite{merkle1987digital} is a cryptographic accumulator, where the accumulator key $ak$ is the hash function, the accumulation value $z$ is the Merkle tree root, and the witness $w$ is the Merkle tree proof. We will use the {\em bilinear accumulator}~\cite{nguyen2005accumulators, kate2010constant} instead of Merkle tree in our protocols, since the witness size of the Merkle tree is logarithmic in the number of values whereas the witness size of the bilinear accumulator is a constant.
On the other hand, the bilinear accumulator requires a trusted dealer, which is a stronger trust assumption than public key infrastructure (PKI).
The trusted dealer needs to know an upper bound on $|\fd|$, i.e., the number of items accumulated
(see the construction in Appendix \ref{app:bilinear}).
In our protocols, $|\fd|$ always equals the number of parties $n$.
Hence, the trusted setup (both the PKI and the accumulator) can be reused across invocations if the parties participating in the extension protocol do not change.
If a trusted dealer for accumulators cannot be assumed, our protocol can use Merkle tree as the accumulator; in that case, the $O(kn^2)$ term in the communication complexity becomes $O(kn^2\log n)$ and our protocol still has an advantage (albeit smaller) over prior art. 



\paragraph{Normalizing the length of cryptographic building blocks.}
Let $\security$ denote the security parameter, $k_h=k_h(\security)$ denote the hash size, $k_s=k_s(\security)$ denote the (multi-)signature size, $k_a=k_a(\security)$ denote the size of the accumulation value and witness of the accumulator.
Further let $k=\max(k_h, k_s, k_a)$; we assume $k=\Theta(k_h)=\Theta(k_s)=\Theta(k_a)=\Theta(\security)$ 
This assumption is reasonable since the signature scheme and accumulator scheme with the shortest output length are both based on pairing-friendly curves, which are believed to require $\Theta(\security)$ bits for $\security$-bit security given the state-of-the-art attack \cite{kim2016extended}. 
As for hash functions, it is common to model them as random oracles, in which case $\lambda$-bit security requires $\Theta(\security)$-bit hash size. 
Therefore, throughout the paper, we can use the same variable $k$ to denote the hash size, signature size and accumulator size for convenience.

\section{Cryptographically Secure Extension Protocols under $t<n/2$ Faults}
\label{sec:crypto}
In this section, we design cryptographically secure extension protocols with improved communication complexity for the synchronous and authenticated setting with $t<n/2$ faults. 
We start by presenting some building blocks that will be frequently used in all our authenticated protocols.
Then, we give an extension protocol for synchronous BA with communication complexity $O(nl+\fa(k)+k n^2)$.
Under synchrony, this also implies a BB protocol with $t<n/2$ and the same communication complexity, by first having the sender send the message to all parties and then performing a Byzantine agreement~\cite{lynch1996distributed}. 
In Appendix \ref{sec:cryptobb:half}, we show another extension protocol for $t<n/2$ BB with communication complexity $O(nl+\fb(k)+\fa(1)+k n^2)$.
The protocols are adapted to the asynchronous case in Section~\ref{sec:async}.
At the end of this section, we discuss the small gap between our BA protocol and a simple lower bound on BA with long messages.

\subsection{Building Blocks: \encode, \distribute and \reconstruct}
\label{sec:building_blocks}

We first define three subprotocols \encode, \distribute and \reconstruct that will be used as building blocks for our cryptographically secure extension protocols, listed in Figure \ref{fig:building_blocks}. 

\begin{itemize}[noitemsep,topsep=0pt]
    \item \encode first divides a message $m$ into $b$ blocks, then compute $n$ coded values $(s_1,...,s_n)$ using RS codes (defined in Section \ref{sec:pre}), and attaches an index $j$ for each value $s_j$. 
    The purpose of \encode is to introduce resilience by encoding the message into fault-tolerant coded values -- after applying \encode to a message $m$, even if $n-b$ coded values in $(s_1,...,s_n)$ are erased, one can recover the message from the remaining coded values.
    
    \item \distribute computes a witness $w_j$ for each indexed value $\pair{j}$ in the input set, and sends the $j$-th value with its witness to party $j$.
    The purpose of \distribute is to distribute the values in a robust yet efficient manner -- if at least one honest party that has the correct message $m$ (the accumulation value $z$ of $m$ is correct) invokes \distribute, then it is guaranteed that any honest party $j$ receives and accepts the $j$-th value $s_j$ of $m$, thanks to the witness $w_j$ sent together with the value.
    
    \item \reconstruct first removes any invalid value $s_j$ that cannot be verified by witness $w_j$ and the accumulation value $z$, and then decode the message $m$ using RS code (defined in Section \ref{sec:pre}) from the remaining values with at most $d_0$ values being removed.
    The purpose of \reconstruct is to recover the message, despite the presence of at most $d_0$ corruptions in the value, which will be detected by the accumulator scheme and thus erased.
\end{itemize}

\begin{figure}[tb]
\begin{mybox}
\begin{itemize}
    \item \textbf{$\encode(m, b)$}
    
    Input: a message $m$, a number $b$
    
    Output: $n$ coded values $s_1,...,s_n$
    
    Divide $m$ into $b$ blocks evenly, $m_1,...,m_{b}$, each has $l/b$ bits where $l$ is the length of $m$. Compute $(s_1,...,s_n)=\texttt{ENC}(m_1,...,m_{b})$ using RS codes, where $\texttt{ENC}$ is defined in Section \ref{sec:tools}.
    Add an index to every value in $(s_1,...,s_n)$, i.e., $\fd=(\pair{1},...,\pair{n})$, and return $\fd$.
    
    \item \textbf{$\distribute(\fd, ak, z)$}
    
    Input: a set of indexed values $\fd=(\pair{1},...,\pair{n})$, an accumulator key $ak$, an accumulation value $z$
    
    Compute $w_j=\createwit(ak, z, \pair{j})$ for every $\pair{j}\in \fd$.
    Send $(s_j, w_j)$ to party $P_j$ for every $j\in [n]$.
    
    \item \textbf{$\reconstruct(\fs, ak, z, d_0)$}
    
    Input: $\fs=((\pair{1}, w_1), ..., (\pair{n}, w_n))$ where each $(\pair{i},w_i)$ is a pair of indexed value and witness, an accumulator key $ak$, an accumulation value $z$, a number $d_0$
    
    Output: a message $m$
    
    For every $j\in [n]$, if $\texttt{Verify}(ak, z', w_j, \pair{j})=\text{false}$, let $s_j=\bot$. 
    Apply $\texttt{DEC}$ on the codewords $(s_1,...,s_n)$ with $c=0$ and $d=d_0$, where $\texttt{DEC}$ is defined in Section \ref{sec:tools}.
    Return $m=m_1|...|m_b$ where $m_1,...,m_b$ are the data returned by \texttt{DEC}.
\end{itemize}
\end{mybox}

\caption{Building Blocks}
\vspace{-8pt}
\label{fig:building_blocks}
\end{figure}

Our extension protocols in Sections \ref{sec:crypto} and \ref{sec:crypto:1-eps} will use \encode at the beginning of the protocol to encode the input message to coded values, use \distribute in the middle to let every party distribute their coded values with the witnesses, and use \reconstruct to reconstruct the original input message after receiving the coded values.

\begin{lemma}\label{lem:crypto:accu_value}
    For any message $m$, 
    let $z=\eval(ak, \encode(m, b))$.
    The adversary cannot find $m' \neq m$ such that $z=\eval(ak, \encode(m', b))$ except for negligible probability in $k$.
\end{lemma}

\begin{proof}
    Let $\fd=\encode(m, b)$ and $\fd'=\encode(m', b)$. 
    By the RS code, the same codewords correspond to the same message. Thus, if $m \neq m'$, we  have $\fd \neq \fd'$, i.e., there exists $d'=\pair{i}$ such that $d \in \fd$ and $d' \not\in \fd'$.
    However, under the accumulation value $z=\eval(ak, \fd'))$, a witness for $d=\pair{i} \not\in \fd'$ exists.
    Due to Lemma~\ref{lem:accumulator}, this happens with negligible in $k$ probability.  
\end{proof}


\subsection{Byzantine Agreement under $<\frac{n}{2}$ faults}
\label{sec:cryptoba:half}

The protocol Synchronous Crypto. $\frac{n}{2}$-BA is presented in Figure \ref{fig:cryptoba:half}.
In the protocol, let $t$ denote the maximum number of Byzantine parties, and let $b=n-t$.
We briefly describe each step of the protocol.
First, each party encodes its message using RS codes and computes the accumulation value for the set of coded values. 
With a deterministic \eval, any honest party with the same accumulator key and set will produce the same accumulation value.
The RS codes can recover the message with up to $t$ coded values being erased, and the accumulation value uniquely corresponds to the set of coded values and equivalently the original message (Lemma \ref{lem:crypto:accu_value}).
Then every party runs an instance of $k$-bit Byzantine agreement with the accumulation value as the input. 
After the above agreement terminates, each party checks whether the agreement output matches its accumulation value, and inputs the result to an $1$-bit Byzantine agreement instance.
If the above agreement outputs $0$, all parties output $\bot$ and abort.
If the above agreement outputs $1$, then at least one honest party has the accumulation value $z_i$ matching with the agreement output $z$, and every honest party will agree on the message corresponding to $z$. Then in \distribute, all parties send the $j$-th coded value to party $P_j$. After that, each honest party $P_j$ will send a valid $j$-th coded value to all other parties, from which the correct message can be obtained in \reconstruct.
One nice property of our protocol is that, if at least one honest party with message $m$ invokes \distribute, then all honest parties can obtain $m$ from \reconstruct (see the proof of Lemma \ref{lem:cryptoba:half:agreement}).
We prove the validity and agreement  properties and analyze the communication complexity below.

\begin{figure}[tb]
\begin{mybox}
    Input of every party $P_i$: An $l$-bit message $m_i$
    
    Primitives: 
    Byzantine agreement oracle, cryptographic accumulator with  \eval, \createwit, \verify
    
    Protocol for party $P_i$: 
    \begin{numlist}
    \item 
    Compute $\fd_i=(\pair{1},...,\pair{n})=\encode(m_i, b)$.
    Compute the accumulation value $z_i=\texttt{Eval}(ak, \fd_i)$.
    Input $z_i$ to an instance of $k$-bit BA oracle.
    \item When the above BA outputs $z$, if $z=z_i$, set $happy_i=1$, otherwise set $happy_i=0$. Input $happy_i$ to an instance of $1$-bit Byzantine agreement oracle.
    \item 
    \begin{itemize}
        \item If the above BA outputs $0$, output $o_i=\bot$ and abort.
        \item If the above BA outputs $1$ and $happy_i=1$, 
        invoke $\distribute(\fd_i, ak, z)$.
    \end{itemize}
    
    \item 
    For the set of pairs $\{(s_i,w_i)\}$ received from the previous step,
    if there exists a pair $(s_i, w_i)$  such that $\texttt{Verify}(ak, z, w_i, \pair{i})=true$, then send $(s_i, w_i)$ to all other parties.
    
    \item 
    If $happy_i=1$, set $o_i=m_i$. Otherwise, let $(s_j, w_j)$ be the message received from party $P_j$ from the previous step and $\fs_i=((\pair{1},w_1),...,(\pair{n},w_n))$,
    and set $o_i=\reconstruct(\fs_i, ak, z, t)$.
    
    \item Output $o_i$.
    
    \end{numlist}
\end{mybox}

\caption{Protocol Synchronous Crypto. $\frac{n}{2}$-BA}
\vspace{-1em}

\label{fig:cryptoba:half}
\end{figure}

\begin{lemma}\label{lem:cryptoba:half:validity}
    If every honest party has the same input message $m_i=m$, all honest parties output the same message $m$.
\end{lemma}

\begin{proof}
    If all honest parties have the same input message $m_i=m$, they compute and input the same accumulation value $z$ to the instance of Byzantine agreement in step $1$.
    Then in step $2$, the BA outputs $z$ by the validity condition, and any honest party sets $happy_i=1$. 
    Therefore, every honest party $P_i$ inputs $1$ to the $1$-bit Byzantine agreement oracle in step $2$. By the validity of the Byzantine agreement oracle, the agreement will output $1$. 
    Then any honest party $P_i$ sets $o_i=m$ in step $5$ since $happy_i=1$. 
    Hence, all honest parties output $m$ when the protocol terminates.
\end{proof}

\begin{lemma}\label{lem:cryptoba:half:agreement}
    All honest parties output the same message.
\end{lemma}

\begin{proof}
    If the Byzantine agreement in step $3$ outputs $0$, then all honest parties output the same message $\bot$.
    If the agreement agreement in step $3$ outputs $1$, then by the validity of the Byzantine agreement, some honest party $P_i$ must input $1$ and thus has $z_i=z$. 
    By Lemma \ref{lem:crypto:accu_value}, any honest party $P_i$ with $happy_i=1$ has the identical message $m$ corresponding to $z$, and sets the output to be $m$ at step $5$.
    In step $3$, any honest party $P_i$ with $happy_i=1$ invokes $\distribute$ to compute witness $w_j$ for each index value $\pair{j}$, and sends the valid $(s_j, w_j)$ pair computed from message $m$ to party $P_j$ for every $P_j$.
    By Lemma \ref{lem:accumulator}, the Byzantine parties cannot generate a different pair $(s_j', w_j')$ that can be verified.
    Therefore, in step $4$, every honest party $P_j$ receives at least one valid $(s_j,w_j)$ pair, and forwards it to all other parties. Since there are at least $b=n-t$ honest parties, in step $5$, each honest party will receive at least $b$ valid coded values. 
    In \reconstruct, using the accumulation value associated with the coded value, any party $P_i$ can detect the corrupted values and remove them. 
    By the property of RS codes, any honest party $P_i$ with $happy_i=0$ is able to recover the message $m$, 
    and any honest party $P_i$ with $happy_i=1$ already has the message $m$.
    Therefore all honest parties outputs $m$.
\end{proof}


\begin{theorem}\label{thm:cryptoba:half}
    Protocol Synchronous Crypto. $\frac{n}{2}$-BA satisfies Termination, Agreement and Validity, and has communication complexity $O(nl+\fa(k)+k n^2)$.
\end{theorem}

\begin{proof}
    
    Termination is clearly satisfied.
    By Lemma \ref{lem:cryptoba:half:agreement}, agreement is satisfied.
    By Lemma \ref{lem:cryptoba:half:validity}, validity is satisfied.
    
    Step $1$ has cost $\fa(k)$, where $k$ is the size of the cryptographic accumulator.
    Step $2$ has cost $\fa(1) \leq \fa(k)$. 
    Step $3$ has cost $O(nl+kn^2)$, since each honest party invokes an instance of \distribute, which leads to an all-to-all communication with each message of size $O(l/b+k)=O(l/n+k)$.
    For step $4$, it also has cost $O(nl+kn^2)$ similarly as step $3$.
    Hence the total cost is $O(nl+\fa(k)+k n^2)$.
\end{proof}

\subsection{Lower Bound on BA for Long Messages}
\label{sec:cryptoba:lb}


Let $\fa(l)$ denote the communication complexity in bits of the best possible deterministic protocol for Byzantine agreement with $l$-bit inputs, $n$ parties, and up to $t=\Theta(n)$ faulty parties.
We show a straightforward lower bound that $\fa(l) = \Omega(nl+\fa(k)+n^2)$ for $l \geq k$ by combining several known lower bounds in the literature.

\begin{theorem}
$\fa(l) = \Omega(nl + \fa(k) + n^2)$ for $l \geq k$.
\end{theorem}

\begin{proof}
The proof combines several simple lower bounds known in the literature. 

First of all, $\fa(l) = \Omega(nl)$ according to \cite{fitzi2006optimally}. 
We briefly mention the proof idea from \cite{fitzi2006optimally} for completeness.
Let a set $A$ of $n-t$ parties have input $m$ and a set $B$ of the rest $t$ parties have input $m'\neq m$.
In scenario $1$, let parties in $B$ be Byzantine but behave as if they are honest. Then by the validity condition, all parties in $A$ will output $m$.
In scenario $2$, let parties in $B$ be honest. To parties in $A$, the scenario $2$ is indistinguishable from scenario $1$, and thus they will output $m$. By the agreement condition, parties in $B$ also need to output $m$.
Therefore each party in $B$ needs to learn the message $m$, which leads to a lower bound on the communication cost of $\Omega(tl)=\Omega(nl)$.

Secondly, since $\fa(l)$ denotes the communication complexity of a BA oracle with $l$-bit inputs, it is clear that $\fa(l) \geq \fa(k)$ for $l \geq k$.

Finally, according to \cite{dolev1982bounds}, $\Omega(n^2)$ is a lower bound on the communication complexity for any deterministic Byzantine agreement protocol tolerating $t=\Theta(n)$ faults (even for single-bit inputs).
Thus, $\fa(l) \geq \fa(1)= \Omega(n^2)$.

The above lower bounds together imply a lower bound $\fa(l) = \Omega(nl+\fa(k)+n^2)$ for deterministic protocol that solves $l$-bit BA.
\end{proof}

By Theorem \ref{thm:cryptoba:half}, our Protocol Synchronous Crypto. $\frac{n}{2}$-BA has cost $O(nl+\fa(k)+k n^2)$, which is very close to the lower bound.
Although it does not meet the lower bound, we remark that further improvements seem challenging.
Notice that if $\fa(k) = \Omega(kn^2)$, then a lower bound of $\Omega(nl+\fa(k)+kn^2)$ follows, matching our upper bound.
Thus, improving upon our upper bound requires a $k$-bit BA oracle whose communication complexity is $o(kn^2)$.

However, if we were to design an $o(kn^2)$ BA protocol, we have to follow a very particular paradigm.
The $\Omega(n^2)$ lower bound from \cite{dolev1982bounds} is a lower bound on the number of messages.
If every message is signed, then $\Omega(kn^2)$ communication must be incurred. 
Yet, we know authentication is necessary for tolerating minority faults.
Thus, such a protocol must use $\Omega(n^2)$ messages in total but only sign a small subset of them. 
We are not aware of any work exploring this direction, and closing this gap is an interesting open problem.


\section{Cryptographically Secure Extension Protocol under $t<(1-\varepsilon)n$}\label{sec:crypto:1-eps}

In this section, we propose an extension protocol for synchronous and authenticated BB with communication complexity $O(nl+\fb(k)+kn^2+n^3)$ under $t<(1-\varepsilon)n$ where $\varepsilon>0$ is some constant.
The protocol still solves Byzantine broadcast under any $t<n$ faults by setting $b=n-t$, but the communication complexity increases by a factor of $1/\varepsilon$ if $\varepsilon$ is not a constant (see Theorem \ref{thm:majority}).
Thus, compared to state-of-art solutions~\cite{ganesh2016broadcast,ganesh2017optimal}, our protocol is more efficient when $\varepsilon$ is a constant but less efficient otherwise.

\begin{figure}[!t]
\begin{mybox}

    Input of the sender $P_s$: An $l$-bit message $m_s$
    
    Primitives: Byzantine broadcast oracle, cryptographic accumulator with \eval, \createwit, \verify
    
    Protocol for party $P_i$: 
    \begin{numlist}
    \item 
    The sender $P_s$ initializes $o_s=m_s, happy_s=1$, and other parties $P_i$ initialize $o_i=\bot, happy_i=0$.
    The sender
    computes $\fd_s=\encode(m_s, b)$, the accumulator value $z_s=\texttt{Eval}(ak, \fd_s)$, and broadcasts $z_s$ by invoking an instance of $k$-bit Byzantine broadcast oracle.
    Let $z_i$ denote the output of the Byzantine broadcast at party $P_i$.
    
    \item For iterations $r=1,..., t+1$: 
    
    \textbf{Distribution step:}
    
    If $happy_i=1$, then sign the \msg message using the multi-signature scheme, send the multi-signature signed by $r$ distinct parties to all other parties, invoke $\distribute(\fd_i, ak, z_i)$, and skip the Distribution step in all future iterations. 
    
    \textbf{Sharing step:}
    
    If a valid $(s_i, w_i)$ pair is received from the Distribution step such that $\verify(ak, z_i, w_i, \pair{i})=\text{true}$, then send $(s_i, w_i)$ to all other parties and skip the Sharing step in all future iterations.
        
    \textbf{Reconstruction step} (no communication involved):
    
    Let $(s_j, w_j)$ be the first message received from party $P_j$ from the Sharing step (possibly from previous iterations). 
    Let $\fs_i=((\pair{1},w_1),...,(\pair{n},w_n))$. Compute $M_i=\reconstruct(\fs_i, ak, z_i, t)$ and $\fd_i=\encode(M_i, b)$.
    If $\texttt{Eval}(ak, \fd_i)=z_i$ and a \msg message signed by $r$ distinct parties excluding $P_i$ was received in the Distribution step of this iteration,
    then set $happy_i=1$, set $o_i=M_i$, and skip the Reconstruction step in all future iterations.
    
    \item Output $o_i$. 
    \end{numlist}

\end{mybox}
\caption{Protocol Synchronous Crypto. $(1-\varepsilon)$-BB}
\vspace{-1em}

\label{fig:cryptoba:1-eps}
\end{figure}

\stitle{Protocol Synchronous Crypto. $(1-\varepsilon)n$-BB}.
The protocol is presented in Figure \ref{fig:cryptoba:1-eps}, and we briefly explain each step of the protocol.
Again let $t$ denote the maximum number of Byzantine parties and let $b=n-t$.
First the sender encodes its message and computes the accumulation value using the coded values. Then the sender broadcasts the accumulation value via an instance of $k$-bit Byzantine broadcast oracle. By the agreement condition, all honest replicas output the same value for BB.
The remaining of the protocol runs in iterations $r=1,2,...,t+1$. 
Each iteration consists $3$ steps.
The Distribution step, Sharing step and Reconstruction step are analogous to steps $3-5$ in Protocol Synchronous Crypto. $\frac{n}{2}$-BA in Figure \ref{fig:cryptoba:half}, but here each step is examined in every iteration for execution, and is executed only once.
The Distribution step aims to distribute the indexed coded values to other parties.
The Sharing step forwards the correct coded value to other parties.
The Reconstruction step aims to reconstruct the original message from the coded values received from other parties and set the output.
Similar to Protocol Synchronous Crypto. $\frac{n}{2}$-BA, the above steps provide a nice guarantee that if at least one honest party with message $m$ invokes \distribute in the Distribution step, then all honest parties can obtain $m$ in the Reconstruction step (see the proof of Lemma \ref{lem:majority:distribute}).

Now we give a more detailed description.
A party becomes happy (i.e., sets $happy_i=1$) if it is ready to output a message that is not $\bot$.
In the first iteration, only the sender is happy;
it invokes \distribute and also signs and sends a message \msg of a constant size.
The role of the message \msg is to be signed by the rest of the parties using multi-signatures to form a signature chain, similar to the Dolve-Strong Byzantine broadcast algorithm \cite{dolev1983authenticated}.
An honest party becomes happy at the end of iteration $r$, if it reconstructs the correct message (matching the agreed upon accumulation value) in the Reconstruction step of iteration $r$ and has received a \msg message signed by $r$ parties in the Distribution step of iteration $r$.
When an honest party becomes happy, it will set its output to be the reconstructed message $M_i$;
then, in the Distribution step of the next iteration (if there is one), it will also send its own signature of \msg to all other parties, and invoke \distribute.
This way, if an honest party becomes happy in the last iteration $r=t+1$, it can be assured that some honest party has invoked \distribute, so that all honest parties will be ready to output the correct message.
We reiterate that each step is executed at most once in the entire protocol.
Finally, after $t+1$ iterations, every party outputs the message.




\begin{lemma}\label{lem:majority:distribute}
    If any honest party $P_i$ invokes \distribute with message $m$, then every honest party $P_j$ outputs $o_j=m$.
\end{lemma}

\begin{proof}
    By the agreement condition of the Byzantine broadcast, the output $z_i$ of the BB at every honest party $P_i$ is identical.
    If an honest party $P_i$ invokes $\distribute$ with message $m$, $m$ satisfies $z_i=\eval(ak, \encode(m,b))$. 
    If any other honest party $P_j$ sets $o_j=m'$ after initialization, it must satisfy $\eval(ak, \encode(m',b))=z_j=z_i$.
    By Lemma \ref{lem:crypto:accu_value}, $m=m'$.
    Thus, we only need to show that every other honest party $P_j$ sets $o_j$.

    Suppose that $P_i$ invokes $\distribute$ in some iteration $r$. According to the subprotocol \distribute, $P_i$ computes a witness $w_j$ for each indexed value $\pair{j}$ and sends the pair $(s_j, w_j)$ to each party $P_j$. 
    According to Lemma \ref{lem:accumulator}, the adversary cannot generate $d'\notin \fd_i$ and a witness $w'$ such that $\verify(ak, z_i, w', d')=true$.
    Then, in Sharing step of iteration $r$, every honest party $P_j$ can identify and forward the valid pair $(s_j,w_j)$ to all other parties, unless it has already done that in previous iterations. 
    Since there are at least $n-t=b$ honest parties, in the Reconstruction step of iteration $r$, every honest party $P_j$ receives at least $n-t=b$ correct coded values. 
    In \reconstruct, using the witness associated with the indexed coded value, every party $P_j$ can identify the corrupted values and remove them. The number of erased values is at most $t$.
    By the property of RS codes, $P_j$ with $happy_j=0$ is able to recover the message $m$.
    
    Furthermore, we will show that each party  receives a \msg message signed by $r$ distinct parties in the Reconstruction step of iteration $r$.
    If $r=1$, then $P_i=P_s$ and every $P_j$ will receive a signature for \msg.
    If $r>1$, then $P_i$ has received a multi-signature of \msg signed by $r-1$ distinct parties excluding $P_i$ in the Reconstruction step of iteration $r-1$;
    $P_i$ adds its own signature of \msg in iteration $r$, so each honest $P_j$ will receive a multi-signature of \msg signed by $r$ distinct parties in the Reconstruction step of iteration $r$.
    
    Therefore, if $happy_j=0$ up till now, then an honest $P_j$ will set $happy_j=1$ and $o_j=m$ in the Reconstruction step of iteration $r$.
    If $happy_j=1$, then $P_j$ has already set $o_j=m$. Note that an honest sender does not set its output again in the Reconstruction step, since the \msg message always contains its signature.
    Once $P_j$ sets $o_j$, it will skip the Reconstruction step in all future iterations, and $o_j$ will not be changed.
    Therefore, all honest parties output $m$ when they terminate.
\end{proof}

\begin{lemma}\label{lem:majority:validity}
    If the sender is honest and has input $m_s$, every honest party outputs $m_s$.
\end{lemma}

\begin{proof}
    In iteration $r=1$, 
    the sender sends a signed \msg to all other parties and invokes \distribute. 
    By Lemma \ref{lem:majority:distribute}, every honest parties output $m_s$.
\end{proof}

\begin{lemma}\label{lem:majority:agreement}
    Every honest party outputs the same message.
\end{lemma}

\begin{proof}
    If all honest parties output $\bot$, then the lemma is true.
    Otherwise, suppose some honest party $P_i$ outputs $o_i=m$ where $m\neq \bot$. 
    If $P_i$ is the sender, then by Lemma \ref{lem:majority:validity}, all honest parties output $m$.
    Now consider the case where $P_i$ is not the sender.
    %
    %
    According to the protocol, if $P_i \neq P_s$ sets $o_i=m \neq \bot$ in the Reconstruction step of iteration $1\leq r\leq t$, $P_i$ will invoke \distribute with $m$ in iteration $r+1$.
    By Lemma \ref{lem:majority:distribute}, all honest parties output $m$.
    If the honest party $P_i$ sets $o_i=m$ in iteration $r=t+1$, according to the protocol, $P_i$ receives a \msg signed by $t+1$ distinct parties. Since there are at most $t$ Byzantine parties, there exists at least one honest party $P_j\neq P_i$ that has signed \msg and invoked \distribute with $o_j=m'$ in a previous iteration $1\leq r'\leq t$. 
    Then, by Lemma \ref{lem:majority:distribute}, all honest parties including $P_i$ output $m'$. Therefore, $m'=m$, and all honest parties output $m$.
\end{proof}

\begin{theorem}\label{thm:majority}
    Protocol Synchronous Crypto. $(1-\varepsilon)n$-BB satisfies Termination, Agreement and Validity.
    The protocol has communication complexity $O(nl/\varepsilon+\fb(k)+kn^2+n^3)$. 
\end{theorem}

\begin{proof}
    Termination is clearly satisfied.
    By Lemma \ref{lem:majority:agreement}, agreement is satisfied.
    By Lemma \ref{lem:majority:validity}, validity is satisfied.
    
    Step $1$ has cost $\fb(k)$ for the $k$-bit BB oracle. 
    The Distribution step has total communication cost $O(nl/\varepsilon+kn^2+n^3)$, since each honest party
    executes the Distribution step at most once, where invoking \distribute has cost $O(n\cdot (l/b+k))=O(\frac{n}{n-t}l+kn)=O(l/\varepsilon+kn)$, and sending the signed \msg message has cost $O((k+n)n)$ where the $(k+n)$ term is due to the signature size and the list of signers in the multi-signature scheme.
    The Sharing step is also performed at most once for every honest party, and has total cost $O(nl/\varepsilon+kn^2)$ since each honest party in the Sharing step sends a message of size $O(l/(n\varepsilon)+k)$ to all other parties.
    The Reconstruction step has no communication cost.
    Hence, the total communication complexity is
    $O(nl/\varepsilon+\fb(k)+kn^2 +n^3)$.
\end{proof}

\stitle{Optimality with the current best BB oracle.}
From Section \ref{sec:related}, the classic Dolev-Strong~\cite{dolev1983authenticated} protocol remains the best deterministic solution for $t>n/2$ BB, with cost $\fb(k) = \Theta((k+\sig)n^2 + n^3)$ for $k$-bit inputs where $\sig$ is the signature size.
Our protocol invokes Dolev-Strong with $k=k_a$ (the size of the accumulation value). Since $\Theta(\sig)=\Theta(k_a)$, our protocol achieves $\fb(l) = O(nl+kn^2+n^3)$.

As before, $\Omega(nl)$ is a trivial lower bound for $l$-bit BB~\cite{fitzi2006optimally} (intuitively, all parties need to receive the sender's message);
in addition $\fb(l) \geq \fb(k)$ if $l \geq k$.
Thus, the $\fb(l) = O(nl+kn^2+n^3)$ communication complexity cannot be further improved unless a deterministic BB protocol better than Dolev-Strong is found.

\section{Cryptographically Secure Extension Protocols Under Asynchrony}
\label{sec:async}

\begin{figure}[tb]
    \begin{mybox}

    Input of every party $P_i$: An $l$-bit message $m_i$
    
    Primitives: 
    asynchronous Byzantine agreement oracle, cryptographic accumulator with \eval, \createwit, \verify
    
    Protocol for party $P_i$: 
    \begin{numlist}
    \item 
    Compute $\fd_i=(\pair{1},...,\pair{n})=\encode(m_i, b)$.
    Compute the accumulation value $z_i=\texttt{Eval}(ak, \fd_i)$.
    Input $z_i$ to an instance of $k$-bit asynchronous Byzantine agreement oracle.
    
    \item When the above ABA outputs $z$, if $z=z_i$, set $happy_i=1$, otherwise set $happy_i=0$. Input $happy_i$ to an instance of $1$-bit asynchronous Byzantine agreement oracle.
    \item 
    \begin{itemize}
        \item If the above ABA outputs $0$, output $o_i=\bot$ and abort.
        \item If the above ABA outputs $1$ and $happy_i=1$, invoke $\distribute(\fd_i, ak, z)$.
    \end{itemize}
    
    \item 
    Wait for a valid $(s_i, w_i)$ pair such that $\verify(ak, z, w_i, \pair{i})=\text{true}$, then send $(s_i, w_i)$ to all other parties.
    
    \item 
    If $happy_i=1$, set $o_i=m_i$.
    Otherwise, perform the following.
    Wait for at least $n-t$ valid pairs $\{(s_j,w_j)\}$ from the previous step that satisfies $\verify(ak, z, w_j, \pair{j})=\text{true}$. Let $\fs_i=((\pair{1},w_1),...,(\pair{n},w_n))$, where $(s_j, w_j)$ is the pair received from party $P_j$.
    Compute $o_i=\reconstruct(\fs_i, ak, z, t)$.
    
    \item Output $o_i$.
    
    \end{numlist}
    \end{mybox}
    \caption{Protocol Asynchronous Crypto. $\frac{n}{3}$-BA}
    \label{fig:asyncba:n/3}
\end{figure}

As mentioned, our cryptographically secure extension protocols can be extended to the asynchronous setting to solve BA and reliable broadcast (RB) under $<n/3$ faults.
No extension protocol has been proposed for this case to the best of our knowledge.
As before, let $t$ denote the maximum number of Byzantine parties, and let $b=n-t$.

\subsection{Asynchronous Byzantine Agreement}
The protocol is presented in Figure \ref{fig:asyncba:n/3}, which consists steps analogous to the synchronous protocol. 
The main difference is that in the asynchronous extension protocol, Steps 4 and 5 are executed once enough messages are received.
As a result, the proofs are also similar to the synchronous version and we omit them. 

\begin{theorem}\label{thm:asyncba}
    Protocol Asynchronous Crypto. $\frac{n}{3}$-BA satisfies Termination, Agreement and Validity, and has communication complexity $O(nl+\fa(k)+k n^2)$.
\end{theorem}


\begin{figure}[tb]

\begin{mybox}

    Input of the sender $P_s$: An $l$-bit message $m_s$
    
    Primitive: 
    asynchronous Byzantine agreement oracle, asynchronous reliable broadcast oracle, cryptographic accumulator with \eval, \createwit, \verify
    
    Protocol for party $P_i$: 
    \begin{numlist}
    \item 
    If $i=s$, perform the following.
    Compute $\fd_s=(\pair{1},...,\pair{n})=\encode(m_s, b)$.
    Compute the accumulation value $z_s=\texttt{Eval}(ak, \fd_s)$.
    Send $m_s$ to every party, and broadcast $z_s$ by invoking a $k$-bit asynchronous reliable broadcast oracle.
    
    \item When receiving the message $m$ from the sender, and the reliable broadcast above outputs $z$, perform the following.
    Compute $\fd_i=(\pair{1},...,\pair{n})=\encode(m, b)$.
    Compute the accumulation value $z_i=\texttt{Eval}(ak, \fd_i)$.
    If $z_i=z$, set $happy_i=1$, otherwise set $happy_i=0$.
    \item If $happy_i=1$, invoke $\distribute(\fd_i, ak, z)$.
    
    \item Step $4$ to $6$ are identical to that of Protocol Asynchronous Crypto. $\frac{n}{3}$-BA in the Figure \ref{fig:asyncba:n/3}, except that the replica computes $\fd_i'=\encode(o_i, b)$, and invokes $\distribute(\fd_i', ak, z)$ at the end of Step $5$.
    
    
    \end{numlist}
        
\end{mybox}
    \caption{Protocol Asynchronous Crypto. $\frac{n}{3}$-RB}
    \label{fig:asyncrb:n/3}
\end{figure}

\subsection{Asynchronous Reliable Broadcast}
\label{app:crypto:async}

Reliable broadcast relaxes the termination property of the broadcast definition (Definition~\ref{defn:BB}):
only when the sender is honest, all honest parties are required to output; otherwise, it is allowed that \emph{either} all honest parties output \emph{or} no honest party outputs.
The agreement property is slightly modified accordingly.

\begin{definition}[Reliable Broadcast]
    A protocol for a set of parties $\fp = \{P_1, ..., P_n\}$, where a distinguished party called the sender $P_s\in \fp$ holds an initial $l$-bit input $m$, is a reliable broadcast protocol tolerating an adversary $A$, if the following properties hold
    \begin{itemize}[noitemsep,topsep=0pt]
        \item Termination. If the sender is honest, then every honest party eventually outputs a message.
        Otherwise, if some honest party outputs a message, then every honest party eventually outputs a message. 
        \item Agreement. If some honest party outputs a message $m'$, then every honest party eventually outputs $m'$. 
        \item Validity. If the sender is honest, all honest parties eventually output the message $m$.
    \end{itemize}
\end{definition}

The extension protocol for asynchronous reliable broadcast is presented in Figure \ref{fig:asyncrb:n/3}.

\begin{lemma}\label{lem:asyncrb}
    If an honest party $P_i$ invokes $\distribute$ with $\fd_i=\encode(m,b)$, then any honest party $P_j$ eventually output $o_j=m$.
\end{lemma}

\begin{proof}
    By the agreement condition of asynchronous reliable broadcast oracle used in step $1$, if any honest party obtains $z$, then any honest party also eventually obtains $z$.
    Then at step $2$, by Lemma \ref{lem:crypto:accu_value}, any honest party $P_j$ with $happy_j=1$ has the identical message $m$ corresponding to $z$, and sets $o_j=m$ at step $5$.
    For other honest parties,
    the honest party $P_i$ with $happy_i=1$ invokes $\distribute$ to compute witness $w_j$ for each indexed value $\pair{j}$, and sends the valid $(s_j, w_j)$ pair computed from message $m$ to party $P_j$ for every $P_j$.    
    By Lemma \ref{lem:accumulator}, the Byzantine parties cannot generate a different pair $(s_j', w_j')$ that can be verified.
    Therefore, in step $4$, every honest party $P_j$ eventually receives at least one valid $(s_j,w_j)$ pair, and forwards it to all other parties. Since there are at least $n-t$ honest parties, in step $5$, each honest party will eventually receive at least $n-t$ valid coded values. 
    In \reconstruct, using the accumulation value associated with the coded value, any party $P_j$ can detect the corrupted values and remove them. 
    By the property of RS codes, any honest party $P_j$ with $happy_j=0$ is able to recover the message $m$, 
    and any honest party $P_j$ with $happy_j=1$ already has the message $m$.
    Therefore all honest parties output $m$.
\end{proof}

\begin{lemma}\label{lem:asyncrb:validity}
    If the sender is honest and has input $m_s$, all honest parties eventually output the same message $m_s$.
\end{lemma}

\begin{proof}
    If the sender is honest, every honest party eventually receive $m_s$, and the asynchronous reliable broadcast eventually outputs the corresponding accumulation value $z_s$ according to the termination condition of asynchronous reliable broadcast. Then in step $2$, any honest party $P_i$ computes a matching accumulation value $z_i=z_s$ with $m_s$,  and sets $happy_i=1$.
    Then any honest party $P_i$ sets $o_i=m$ in step $5$ since $happy_i=1$.
    Hence, all honest parties output $m$ when the protocol terminates.
\end{proof}

\begin{lemma}\label{lem:asyncrb:agreement}
    If some honest party outputs a message $m$, then every honest party eventually outputs $m$.
\end{lemma}

\begin{proof}
    Suppose any honest party $P_i$ outputs $o_i=m$.
    If $P_i$ has $happy_i=1$ at step $5$, it invokes \distribute at step $3$. Then by Lemma \ref{lem:asyncrb}, all honest parties eventually output $m$.
    If $P_i$ has $happy_i=0$ at step $5$, then it reconstructs the message $m$ from \reconstruct, and invokes \distribute. Then by Lemma \ref{lem:asyncrb}, all honest parties eventually output $m$. 
\end{proof}

\begin{theorem}\label{thm:asyncrb}
    Protocol Asynchronous Crypto. $\frac{n}{3}$-RB satisfies Termination, Agreement and Validity.
    The protocol has communication complexity $O(nl+\fb(k)+k n^2)$.
\end{theorem}

\begin{proof}
    Termination is proved by Lemma \ref{lem:asyncrb:validity} and \ref{lem:asyncrb:agreement}. Agreement is proved by Lemma \ref{lem:asyncrb:agreement}. Validity is proved by Lemma \ref{lem:asyncrb:validity}.
    
    Step $1$ has cost $O(nl+\fb(k))$, where $k$ is the size of the cryptographic accumulator.
    Step $3$ and $5$ in total have cost $O(nl+kn^2)$, since each honest party invokes at most one instance of \distribute, which leads to an all-to-all communication with each message of size $O(l/b+k)=O(l/n+k)$.
    For step $4$, it also has cost $O(nl+kn^2)$.
    Hence the total cost is $O(nl+\fb(k)+k n^2)$.
\end{proof}


\section{Conclusion}
We investigate and propose several extension protocols with improved communication complexity for solving Byzantine broadcast and agreement under various settings.
We propose simple yet efficient authenticated extension protocols with improved communication complexity, for Byzantine agreement under $t<n/2$, and for Byzantine broadcast under $t< (1-\varepsilon)n$ where $\varepsilon>0$ is a constant.
The above results can be extended to the asynchronous case to obtain authenticated extension protocols for Byzantine agreement and reliable broadcast.



\bibliography{sample-base}

\appendix

\appendix

\section{The Bilinear Accumulator}
\label{app:bilinear}

To satisfy our assumption that output lengths of cryptographic primitives are the same order, we  choose an accumulator implementation called the bilinear accumulator \cite{nguyen2005accumulators}. 

\textbf{Bilinear Pairing.}
Let $\mathbb{G}_1, \mathbb{G}_2$ be two cyclic multiplicative groups of prime order $p$. Let $g_1,g_2$ be the corresponding generator, and there exists an isomorphism $\phi:\mathbb{G}_2\rightarrow \mathbb{G}_1$ such that $\phi(g_2)=g_1$.
Let $\mathbb{G}_M$ also be a cyclic multiplicative group of prime order $p$, and $e:\mathbb{G}_1\times \mathbb{G}_2\rightarrow \mathbb{G}_M$ is a bilinear pairing if satisfies the following properties:
\begin{enumerate}
    \item Bilinearity: $e(P^a, Q^b)=e(P,Q)^{ab}$ for all $P\in \mathbb{G}_1$, $Q\in \mathbb{G}_2$ and $a,b\in \mathbb{Z}_p$;
    \item Non-degeneracy: $e(g_1,g_2)\neq 1$;
    \item Computability: There is an efficient algorithm to compute $e(P,Q)$ for all $P\in \mathbb{G}_1$ and $Q\in \mathbb{G}_2$.
\end{enumerate}

\textbf{Accumulator Construction.}
For accumulator construction, we can have $\mathbb{G}_1=\mathbb{G}_2=\mathbb{G}$ and $g_1=g_2=g$. 
The bilinear accumulator works for elements in $\mathbb{Z}_p^*$ and the accumulation value is an element in $\mathbb{G}$.
Therefore, we assume a function $f:\mathcal{U}\rightarrow \mathbb{Z}_p^*$ that maps any value in the input domain $\mathcal{U}$ to an value in $\mathbb{Z}_p^*$.
Let $\fd=\{d_1,...,d_n\}$ be a set of $n$ values in $\mathbb{Z}_p^*$ after applying the function $f$. 
Let $s$ denote a trapdoor that is hidden from all parties participating in the extension protocol.
Before the extension protocol starts, all parties obtain a set of Strong Diffie-Hellman ($q$-SDH) public parameters $\langle g^s, g^{s^2},..., g^{s^n} \rangle$, via the trusted setup. We assume a trusted dealer that generates the trapdoor $s$ and distributes the parameters $\langle g^s, g^{s^2},..., g^{s^n} \rangle$ to all parties.
Let $C_{\fd}(x)=(x+d_1)(x+d_2)\cdots(x+d_n)$ denote the characteristic polynomial of $\fd$ with coefficients $c_0,c_1,...,c_n$, so that
$C_{\fd}(x)=(x+d_1)(x+d_2)\cdots(x+d_n)=c_0+c_1x+...+c_nx^n$.
Let $q_i(x)=\frac{C_{\fd}(x)}{x+d_i}=\prod_{j\neq i}(x+d_j)=c_0^{(i)}+c_1^{(i)}x+...+c_{n-1}^{(i)}x^{n-1}$ denote the quotient polynomial.

\begin{itemize}
    \item $\gen(1^k,n)$ returns a uniformly random tuple $ak=(p,\mathbb{G},\mathbb{G}_M,e,g)$ of bilinear pairings parameters, where $p$ is of size $k$.
    
    \item $\eval(ak, \fd)$ computes and returns the accumulation value as $z=g^{C_{\fd}(s)}=g^{c_0+c_1s+...+c_ns^n}=g^{c_0}(g^s)^{c_1}\cdots (g^{s^n})^{c_n}$.
    
    \item $\createwit(ak, z, d_i)$ computes and returns the witness $w_i$ as
    $w_i=g^{\frac{C_{\fd}(x)}{s+d_i}}=g^{c_0^{(i)}+c_1^{(i)}s+...+c_{n-1}^{(i)} s^{n-1}}=g^{c_0^{(i)}}(g^s)^{c_1^{(i)}}\cdots (g^{s^{n-1}})^{c_{n-1}^{(i)}}$ if $d_i\in \fd$, and $w_i=\bot$ if $d_i\notin \fd$.
    
    \item $\verify(ak, z, w_i, d_i)$ tests whether
    $e(g^{d_i}\cdot g^{s}, w_i)=e(z, g)$ where $e$ is the bilinear pairing, and return the result.
\end{itemize}

Since $p$ is of size $k$, and the accumulation value $z$ or any witness $w_i$ is an element in the group $\mathbb{Z}_p$, they all have size $k$ bits.
It is believed that this gives $\Theta(k)$ bits of security \cite{kim2016extended}.

\section{An Alternative Cryptographically Secure Extension Protocol for Byzantine Broadcast}
\label{sec:cryptobb:half}

In this section, we present an alternative cryptographically secure extension protocol for synchronous and authenticated Byzantine broadcast under synchrony and asynchrony. It is very similar to that of synchronous Byzantine agreement, as presented in Figure \ref{fig:cryptobb:half}.
The main difference is that the sender uses a $k$-bit Byzantine broadcast oracle to broadcast its accumulation value, instead of every party inputting its accumulation value to a $k$-bit Byzantine agreement oracle.
As a result, the communication complexity becomes $O(nl+\fb(k)+\fa(1)+k n^2)$ instead of $O(nl+\fa(k)+k n^2)$.

Note that $\fb(k) \leq nl + \fa(k)$ given the simple BB-to-BA reduction (i.e., having the sender sending its value and invoking BA).
Therefore, the protocol in this section is asymptotically no worse than the one in the main body.

\begin{figure}[h!]
\begin{mybox}

    Input of the sender $P_s$: An $l$-bit message $m_s$
    
    Primitives: Byzantine broadcast oracle, 
    Byzantine agreement oracle, cryptographic accumulator with \eval, \createwit, \verify
    
    Protocol for party $P_i$: 
    \begin{numlist}
    \item 
    If $i=s$, perform the following.
    Compute $\fd_s=(\pair{1},...,\pair{n})=\encode(m_s, b)$.
    Compute the accumulation value $z_s=\texttt{Eval}(ak, \fd_s)$.
    Send $m_s$ to every party, and broadcast $z_s$ by invoking a $k$-bit Byzantine broadcast oracle.
    
    \item When receiving the message $m$ from the sender, and the Byzantine broadcast above outputs $z$, perform the following.
    Compute $\fd_i=(\pair{1},...,\pair{n})=\encode(m, b)$.
    Compute the accumulation value $z_i=\texttt{Eval}(ak, \fd_i)$.
    If $z_i=z$, set $happy_i=1$, otherwise set $happy_i=0$.
    Input $happy_i$ to an instance of single-bit Byzantine agreement oracle.
    \item Steps $3$ to $6$ are identical to that of Protocol Synchronous Crypto. $\frac{n}{2}$-BA.
    \end{numlist}

\end{mybox}
\caption{Protocol Synchronous Crypto. $\frac{n}{2}$-BB}
\label{fig:cryptobb:half}
\end{figure}

\begin{lemma}\label{lem:cryptobb:half:validity}
    If the sender is honest and has input $m_s$, all honest parties output the same message $m_s$.
\end{lemma}

\begin{proof}
    If the sender is honest, every honest party will receive $m_s$ and the Byzantine broadcast outputs the corresponding accumulation value $z_s$. Then in step $2$, any honest party $P_i$ computes a matching accumulation value $z_i=z_s$ with $m_s$, sets $happy_i=1$ and inputs $1$ to an $1$-bit Byzantine agreement oracle. By the validity of the Byzantine agreement, the agreement will output $1$. 
    Then any honest party $P_i$ sets $o_i=m$ in step $5$ since $happy_i=1$. 
    Hence, all honest parties output $m$ when the protocol terminates.
\end{proof}

\begin{lemma}\label{lem:cryptobb:half:agreement}
    All honest parties output the same message.
\end{lemma}

\begin{proof}
The proof is identical to that of Lemma \ref{lem:cryptoba:half:agreement}.
\end{proof}

\begin{theorem}\label{thm:cryptobb:half}
    Protocol Synchronous Crypto. $\frac{n}{2}$-BB satisfies Termination, Agreement and Validity.
    The protocol has communication complexity $O(nl+\fb(k)+\fa(1)+k n^2)$.
\end{theorem}

\begin{proof}
    Termination is clearly satisfied.
    By Lemma \ref{lem:cryptobb:half:agreement}, agreement is satisfied.
    By Lemma \ref{lem:cryptobb:half:validity}, validity is satisfied.
    
    Step $1$ has cost $O(nl+\fb(k))$, where $k$ is the size of the cryptographic accumulator.
    Step $2$ has cost $O(\fa(1))$. 
    Step $3$ has cost $O(nl+kn^2)$, since each honest party invokes an instance of \distribute, which leads to an all-to-all communication with each message of size $O(l/b+k)=O(l/n+k)$.
    For step $4$, it also has cost $O(nl+kn^2)$.
    Hence the total cost is $O(nl+\fb(k)+\fa(1)+k n^2)$.
\end{proof}

\section{Error-free Extension Protocols}\label{sec:errorfree}

In this section, we present improved error-free extension protocols for synchronous Byzantine agreement/broadcast, under synchrony (Section \ref{sec:errorfree:sync}) and asynchrony (Section \ref{app:errorfree:async}).
We consider computationally unbounded adversary in this section. 
Without loss of generality, we assume the number of Byzantine parties is $t=\lfloor \frac{n-1}{3}\rfloor$ in this section.

\begin{table}[tb] \centering
\resizebox{\columnwidth}{!}{

\begin{tabular}{|l|l|l|l|l|}
\hline
\textbf{Model}   
& \textbf{Problem}                      
& \begin{tabular}[c]{@{}l@{}}\textbf{Communication} \\ \textbf{Complexity}\end{tabular}            
& \begin{tabular}[c]{@{}l@{}}\textbf{Input range} $l$ \\ \textbf{to reach optimality}\end{tabular}     
& \textbf{Reference} \\ 
\hline
sync. & 
agreement/broadcast
& 
\begin{tabular}[c]{@{}l@{}} 
$O(nl+n^2\fb(1))$\\ 
$O(nl+n\fb(1)+n^3)$ 
\end{tabular} & 
\begin{tabular}[c]{@{}l@{}} 
$\Omega(n^3)$\\ 
$\Omega(n^2)$
\end{tabular} & 
\begin{tabular}[c]{@{}l@{}} 
\cite{ganesh2017optimal, ganesh2016broadcast}\\ 
\textbf{This paper} 
\end{tabular}
 \\
\hline 
async. & 
\begin{tabular}[c]{@{}l@{}}
reliable broadcast\\ reliable broadcast\\ agreement\\agreement
\end{tabular} 
& 
\begin{tabular}[c]{@{}l@{}}
$O(nl+n^2\log n\fb(1))$ \\ 
$O(nl+n\fb(1)+n^3\log n)$ \\ 
$O(nl+n^3\log n\fb(1)+n\fa(1))$ \\ 
$O(nl+n^2\fb(1)+n\fa(1)+n^4\log n)$ 
\end{tabular}
&
\begin{tabular}[c]{@{}l@{}} 
$\Omega(n^3\log n)$\\
$\Omega(n^2\log n)$\\
$\Omega(n^4\log n)$\\
$\Omega(n^3\log n)$
\end{tabular}   &
\begin{tabular}[c]{@{}l@{}} 
\cite{patra2011error,ganesh2017optimal}\\ 
\textbf{This paper}\\
\cite{patra2011error,ganesh2017optimal}\\ 
\textbf{This paper}
\end{tabular}  \\
\hline

\end{tabular}
}
\caption{Error-free Extension Protocols for Byzantine Agreement and Broadcast under $t<n/3$
}\label{table:summary:errorfree}

\end{table}

\subsection{Synchronous Error-free Extension Protocols under $t<\frac{n}{3}$ Faults}\label{sec:errorfree:sync}

For error-free multi-valued synchronous Byzantine agreement, the state-of-art extension protocol has communication complexity $O(nl+n^2\fb(1))=O(nl+n^4)$ \cite{patra2011error, ganesh2017optimal}. 
Here, we propose a protocol under the same setting with improved communication complexity $O(nl+n^3+n\fb(1))=O(nl+n^3)$.
For $l\geq O(n^2)$, our protocol is optimal in communication complexity, and for $l\leq O(n^2)$, our protocol has communication cost $O(n^3)$, which is better than the previous work \cite{patra2011error, ganesh2017optimal} by a factor of $n$.
For the synchronous case, Byzantine broadcast can be constructed from Byzantine agreement with the same asymptotic communication complexity, by first letting the sender send the message to all parties and then perform a Byzantine agreement to reach agreement \cite{lynch1996distributed}. 
Therefore, we only present the Byzantine agreement protocol.

\medskip\stitle{Building Block: $\texttt{STAR}$ protocol \cite{ben1993asynchronous}.}
The following building block is adopted from \cite{ganesh2017optimal}.
A sub-protocol (Figure \ref{fig:star}) called $\texttt{STAR}$ \cite{ben1993asynchronous} is used to find an $(n,t)$-star in a given undirected graph.
\begin{definition}[$(n,t)$-star \cite{ben1993asynchronous}]
    For a given undirected graph $G=(\fp,E)$, an $(n,t)$-star is a pair $(\fc, \fd)$ of sets with $\fc\subseteq \fd\subseteq\fp$ that satisfies 
    \begin{enumerate}
        \item $|\fc|\geq n-2t$, $|\fd|\geq n-t$
        \item There exists an edge $(P_i,P_j)\in E$ for $\forall P_i\in \fc, P_j\in \fd$
    \end{enumerate}
\end{definition}


\begin{figure}[tb]
\begin{mybox}
\begin{numlist}
    \item Let $G=(\fp, E)$ be the input graph. 
    Let $H=\overline{G}=(\fp, \overline{E})$ be the complementary graph of $G$.
    Find a maximum matching in $H$ using any deterministic algorithm such as \cite{blum1990new}. Let $M$ be the matching and $N$ be the set of matched nodes. Let $\overline{N}=\fp \backslash N$.
    \item Compute output as follows:
    \begin{numlist}
        \item Find the set $T=\{ P_i\in \overline{N} ~|~ \exists P_j,P_k \text{ s.t. } (P_j,P_k)\in M \text{ and } (P_i,P_j),(P_i,P_k)\in \overline{E} \}$.
        Let $\fc = \overline{N}\backslash T$.
        \item Find the set $B\subseteq N$ of matched nodes that have neighbors in $\fc$ in $H$. 
        That is, $B=\{P_j\in N~|~ \exists P_k\in \fc \text{ s.t. } (P_j,P_k)\in \overline{E}\}$.
        Let $\fd = \fp \backslash B$.
        \item If $|\fc|\geq 2n-t$ and $|\fd|\geq n-t$, output $(\fc, \fd)$. Otherwise, output $\texttt{noSTAR}$.
    \end{numlist}
\end{numlist}

\end{mybox}

\caption{Protocol $\texttt{STAR}$}

\label{fig:star}
\end{figure}

\stitle{Protocol Synchronous Error-free $\frac{n}{3}$-BA.}
Now we present an improved extension protocol for error-free multi-valued synchronous Byzantine agreement when $t<n/3$, as presented in Figure \ref{fig:syncba}.
The algorithm is heavily inspired by the error-free protocol from \cite{patra2011error, ganesh2017optimal}, and has a similar structure. 
Here we briefly describe each step of the protocol.
Initially, each party divides and encodes its message into $n$ blocks via RS code, and sends blocks to the corresponding party. Then each party compares its block and the block received from others, and constructs a vector to record whether the corresponding blocks are identical.
After all the parties exchange their vectors, each party $P_i$ constructs a graph $G_i$ from which a set $\fe_i$ of parties is derived. Then each party $P_i$ broadcasts whether it has successfully obtained $\fe_i$, and sends $\fe_i$ to all other parties. When there are enough parties successfully obtaining the set, each honest party can extract a correct piece of codeword $maj$ and send to others. From all the codewords received, each honest party is able to reconstruct the message and thus reach an agreement.

\begin{figure}[!t]
\begin{mybox}

    Input of every party $p_i$: An $l$-bit message $m_i$
    
    Primitives: Broadcast oracle for a single bit, $\texttt{STAR}$
    
    Protocol for party $P_i$: 
    \begin{numlist}
    \item Divide the $l$-bit message $m_i$ into $t+1$ blocks, $m_{i0},\cdots,m_{it}$, each has 
    $l/(t+1)$ bits. Compute $(s_{i1},\cdots, s_{in})=\texttt{ENC}(m_{i0},\cdots,m_{it})$. 
    Send $s_{ii}$ to every party. Send $s_{ij}$ to $P_j$ for $j=1,\cdots, n$.
    
    \item Construct a binary vector $v_i$ of length $n$.
    Assign $v_i[j] = 1$, if $s_{ij} = s_{jj}$ and $s_{ii} = s_{ji}$ where $s_{jj}$ and $s_{ji}$ are received from $P_j$. 
    Otherwise assign $v_i[j] = 0$. Send $v_i$ to every party.
    
    \item Construct an undirected graph $G_i$ with parties in $\fp$ as vertices, and add an edge $(P_x,P_y)$ if $v_x[y]=v_y[x]=1$. Invoke $\texttt{STAR}(G_i)$.
    \item
    If $(\fc_i, \fd_i)$ is returned by $\texttt{STAR}$, find $\ff_i$ of size at least $2t+1$ as the set of parties who have at least $t+1$ neighbours in $\fc_i$ in graph $G_i$, and find $\fe_i$ of size at least $2t+1$ as the set of parties who have at least $2t+ 1$ neighbours in $\ff_i$ in graph $G_i$.
    Any party $P_j$ is viewed as its neighbor for the purpose of finding $\fe_i$.
    Obtain the set $\fe_i$ as above if possible, otherwise let $\fe_i=\emptyset$. 
    
    \item Broadcast a single bit of $1$ using the single-bit Byzantine broadcast primitive if $\fe_i\neq \emptyset$. 
    Send $\fe_i$ as an $n$-bit vector to every party.
    Otherwise, broadcast a single bit of $0$ using the single-bit Byzantine broadcast primitive.
    
    \item After the above Byzantine broadcasts finish, perform the following.
    \begin{itemize}
        \item If the above Byzantine broadcast delivers $\geq 2t+1$ $1$'s:
        let $\mathbf{E}$ contain the corresponding set of $\fe$'s that are received by $P_i$.
        For each $\fe_x\in \mathbf{E}$, let $maj_x$ be the value $s_{ji}$ received from the majority of
        the parties in $\fe_x$. 
        That is, $maj_x$ satisfies that $|\{j\in \fe_x ~|~ s_{ji}=maj_x \}|\geq \lceil (\fe_x+1)/2 \rceil$.
        If such a majority does not exist, let $maj_x=\perp$.
        Find a subset $\mathbf{E'}\subseteq \mathbf{E}$, such that $|\mathbf{E'}|\geq t+1$ and for any $\fe_x, \fe_y\in \mathbf{E'}$, $maj_x=maj_y\neq \perp$. Denote the above value as $maj$.
        Send the value $maj$ to every party.
        \item If the above Byzantine broadcast delivers $< 2t+1$ $1$'s:
        agree on some predefined message $m'$ of length $l$ and abort.
    \end{itemize}
        
    \item 
    Let $(maj_1, \cdots , maj_n)$ be the vector where $maj_j$ is received from $P_j$ in the above step. 
    Apply $\texttt{DEC}$ on $(maj_1, \cdots , maj_n)$ with $c = t$ and $d = 0$. 
    Let $m_0, m_1, \cdots , m_t$ be the data returned by $\texttt{DEC}$. Output $m = m_0| \cdots |m_t$.
    \end{numlist}
\end{mybox}
\caption{Protocol Synchronous Error-free $\frac{n}{3}$-BA}
\label{fig:syncba}
\end{figure}

Comparing to the protocol in \cite{patra2011error, ganesh2017optimal}, the novelty of our protocol is that, instead of every party $P_i$ broadcasting a length $n$ vector $v_i$ via Byzantine broadcast in step $2$ which will result in communication complexity $O(n^2\fb(1))=O(n^4)$, we only let each party to send the vector to all other parties and thus reduce the cost to $O(n^3)$.
As a result, different parties may receive different vectors from the Byzantine parties, which leads to different constructions of set $\fc_i,\fd_i,\ff_i,\fe_i$ at different party $P_i$ instead of the identical sets $\fc,\fd,\ff,\fe$ at all parties as in Protocol $\frac{n}{3}$-BA from \cite{patra2011error, ganesh2017optimal}.
How to resolve such conflicting information and still obtain useful information to reconstruct an identical message at all honest parties is the main contribution of our protocol.
The observation is that, as will be shown later in the proofs, honest parties are the majority in any set $\fe_i$ and all have the same message, which can be used to extract enough pieces of identical codewords for reconstructing the message. The above procedure is done in step $6$ and $7$.
Recall that the function \texttt{DEC} is capable to tolerate up to $c$ errors and $d$ erasures in the codewords, if and only if $n-(t+1)\geq 2c+d$.


\begin{lemma}\label{lem:1}
    For any honest party $P_i$, if it obtains nonempty set $\fe_i$ in step $4$ of the protocol, then the honest parties in $\fe_i$ hold the same message of length $l$.
\end{lemma}

\begin{proof}
    We show the following properties are satisfied if $P_i$ is able to obtain its $\fe_i$ in step $3$.
    \begin{itemize}
        \item The honest parties in $\fc_i$ hold the same message of length $l$.
        
        By the definition of $(n,t)$-star, $|\fd_i|\geq n-t$ and any two parties $P_j\in \fc_i, P_k\in \fd_i$ are connected by an edge. This implies that $\fd_i$ contains at least $n-2t\geq t+1$ honest parties $\{P_{i_1},P_{i_2},...,P_{i_q}\}$, and every honest party $P_j\in \fc_i$ connects to all those honest parties. Then for any honest party $P_j\in \fc_i$, by definition we have $s_{ji_x}=s_{i_x i_x}$ for all $x=1,2,...,q$ where $q\geq t+1$. This means that the codewords of each honest party in $\fc$ have at least $t+1$ elements in common. Since the codewords used in the protocol are $(n,t+1)$ RS code, all honest parties in $\fc_i$ hold the same message of length $l$. Let the common message be $m$, and let $(s_1,...,s_n)=\texttt{ENC}(m_0,m_1,...,m_t)$ where $m=m_0|m_1|...|m_t$. 
        \item Every honest party $P_j\in \ff_i$ holds $s_j$.
        
        Recall every honest party $P_j\in \ff_i$ has at least $t+1$ neighbors in $\fc_i$, and therefore at least $1$ honest neighbor $P_k$ in $\fc_i$. Since $P_j,P_k$ are neighbors, $s_{jj}$ of $P_j$ equals $s_{kj}$ of $P_k$. Since $P_k$ holds message $m$, $s_{kj}=s_{j}$, which implies that $P_j$ holds $s_j$.
        
        \item The honest parties in $\fe_i$ hold the same message of length $l$.
        
        Recall every honest party $P_i\in \fe$ has at least $2t+1$ neighbors in $\ff$, and therefore at least $t+1$ honest neighbors $\{P_{i_1},P_{i_2},...,P_{i_q}\}$ in $\ff$ where $q\geq t+1$. 
        Since $P_i$ and $P_{i_x}$ are connected, $s_{i i_x}=s_{i_x i_x}$ for $1\leq x\leq q$.
        Recall that $s_{i_x i_x}=s_{i_x}$.
        Therefore the codewords of $P_i$ has at least $t+1$ elements identical to the elements of $(s_1,...,s_n)$, where $(s_1,...,s_n)=\texttt{ENC}(m_0,m_1,...,m_t)$ where $m=m_0|m_1|...|m_t$. 
        Since the codewords used in the protocol are $(n,t+1)$ RS code, all honest parties in $\fe_i$ hold the same message of length $l$. \qedhere
    \end{itemize}
\end{proof}

\begin{lemma}\label{lem:2}
    If all honest parties start with the same input $m$, then every honest party $P_i$ will obtain a set $\fe_i\neq \emptyset$.
\end{lemma}

\begin{proof}
    When all honest parties have the same input $m$, they will generate the same codewords. This implies all honest parties will connect to each other in $G_i$, which forms a clique of size $\geq 2t+1$. 

    First we show that $\fc_i$ contains at least $t+1$ honest parties. Recall that in the protocol $\texttt{STAR}$, $\fc_i=(\fp\backslash N)\backslash T$, where $N$ is the set of matched parties in the complementary graph of $G$, and $T$ is the set of parties that connects to both endpoints of a matching.  For any honest party $P_j\in N$ such that $(P_j,P_k)\in M$, it is ensured that $P_k$ is Byzantine since $P_j,P_k$ have conflicting messages. 
    Similarly, for any honest party $P_j\in T$ such that $(P_x,P_y)\in M$, $(P_j,P_x)\in \overline{E}$ and $(P_j, P_y)\in \overline{E}$, both $P_x$ and $P_y$ are ensured to be Byzantine since they have conflicting messages with $P_j$.
    Therefore, for any honest party excluded from $\fc_i=(\fp\backslash N)\backslash T$, there exists at least one corresponding Byzantine party excluded from $\fc_i$ as well. Since there are at most $t$ Byzantine parties, at most $t$ honest parties are excluded from $\fc_i$. Hence $\fc_i$ contains at least $t+1$ honest parties.
    Since $\fc_i$ contains at least $t+1$ honest parties, $\ff_i, \fe_i$ will subsequently contain all honest parties and have size $\geq 2t+1$.
\end{proof}


\begin{lemma}\label{lem:3}
    All honest parties output the same message.
\end{lemma}
\begin{proof}
    After the reliable broadcast of a single bit, all honest parties will deliver an identical set of $\{0,1\}$. Therefore if one honest party delivers $<2t+1$ $1$'s and outputs the predefined message $m'$, all honest parties will output $m'$.

    Now consider the case where all honest parties deliver $\geq 2t+1$ $1$'s for the reliable broadcast of a single bit. This implies that at least $t+1$ honest parties $P_i$ send the set $\fe_i$ to all parties. Denote the above set of honest parties as $\mathcal{H}$. By Lemma \ref{lem:1}, we have all honest parties in $\fe_i$ have the same message $m$ of length $l$. 
    Also, for any two honest parties $P_i,P_j$ above and their set $\fe_i,\fe_j$, we know that $\fe_i\cap\fe_j$ contains at least one honest party, since $|\fe_i|\geq 2t+1$ and  $|\fe_j|\geq 2t+1$. Both facts above imply that all honest parties in $\fe_i\cup \fe_j$ have the same message $m$.

    Since at least $t+1$ honest parties send their $\fe$ set to all parties, the algorithm can find the feasible set $\mathbf{E'}$ that contains all honest parties that are in $\mathcal{H}$. 
    The conditions ``$|\mathbf{E'}|\geq t+1$, for any $\fe_x, \fe_y\in \mathbf{E'}$, $maj_x=maj_y\neq \perp$'' in step $6a$ can be satisfied: 
    Honest parties are the majority in any $\fe_x\in \mathbf{E'}$ since $|\fe_x|\geq 2t+1$, and they have the same message $m$. Thus for any $\fe_x\in \mathbf{E'}$, $maj_x$ is the same.
    Let $(s_1,...,s_n)=\texttt{ENC}(m_0,m_1,...,m_t)$ where $m=m_0|m_1|...|m_t$. 
    We know that $maj_x=s_x$ for all $\fe_x\in \mathbf{E'}$. 
    This implies that $maj_j=s_j$ for all honest party $P_j$ in step $7$.
    Hence at least $2t+1$ values received by any honest party in step $7$ are identical to the corresponding elements in $(s_1,...,s_n)$.
    Since the codewords used in the protocol are $(n,t+1)$ RS code which corrects at most $t$ failures, after step $7$, all honest parties will recover and output the same message $m$.
\end{proof}

\begin{lemma}\label{lem:4}
    If all honest parties start with the same input $m$, then all honest parties output $m$.
\end{lemma}

\begin{proof}
    When all honest parties start with the same input $m$, by Lemma \ref{lem:2}, every honest party $P_i$ will obtain its set $\fe_i\neq \emptyset$.  Then, all honest parties will broadcast a single bit of $1$, and deliver at least $2t+1$ $1$'s for the broadcast. By the same proof of Lemma \ref{lem:3}, all honest parties will output the same message $m$.
\end{proof}

\begin{theorem}
    Protocol Synchronous Error-free $\frac{n}{3}$-BA satisfies Termination, Agreement and Validity.
    The protocol has communication complexity $O(nl+n^3+n\fb(1))$.
\end{theorem}

\begin{proof}
    Termination is clearly satisfied.
    Agreement is proved by Lemma \ref{lem:3}, and Validity is proved by Lemma \ref{lem:4}.

    Step $1$ has communication cost $O(n^2l/(t+1))=O(nl)$. Step $2$ has communication cost $O(n^3)$. Step $5$ has communication cost $O(n^3+n\fb(1))$. 
    Step $6$ has communication cost $O(nl)$.
\end{proof}


\subsection{Asynchronous Error-free Extension Protocols under $t<\frac{n}{3}$ Faults}
\label{app:errorfree:async}

We can extend Protocol Synchronous Error-free $\frac{n}{3}$-BA from the previous Section to an asynchronous reliable broadcast, as presented in Figure \ref{fig:asyncrb:error_free}.
For brevity, we only present the difference.
Since the system is asynchronous and the channel is reliable, each party can only expect to receive the messages from honest parties eventually. Thus, instead of receiving all vectors and then constructing the sets as in the step $3,4$ of the synchronous protocol, each party can only try to construct the sets every time a new message is received as in step $3, 4$.
Similar to the synchronous protocol, the parties can obtain enough pieces of correct codewords to reconstruct the message as in step $6$ and $7$.

\begin{figure}[tb]

\begin{mybox}

    Input of the sender $P_s$: An $l$-bit message $m_s$
    
    Primitive: Asynchronous reliable broadcast oracle for a single bit, $\texttt{STAR}$
    
    Protocol for party $P_i$: 
    \begin{numlist}
    \item If $i=s$, send $m_s$ to every party. Wait until receiving the message $m_i$ from the sender.
    \item Same as step $1$ of Protocol Synchronous Error-free $\frac{n}{3}$-BA. 
    
    \item When receiving $s_{jj}$ and $s_{ji}$ from $P_j$, send $\texttt{OK}(P_i,P_j)$ to every party if $s_{ij} = s_{jj}$.
    Construct an undirected graph $G_i$ with parties in $\fp$ as vertices. Add an edge $(P_x,P_y)$ every time when $\texttt{OK}(P_x,P_y)$ is received from $P_x$ and $\texttt{OK}(P_y,P_x)$ is received from $P_y$. 
    If the edge $(P_x,P_y)$ is new,
    invoke $\texttt{STAR}(G_i)$.
    
    \item Same as step $4$ of Protocol Synchronous Error-free $\frac{n}{3}$-BA.
    \item 
    If the set $\fe_i$ is obtained for the first time, broadcast a single bit of $1$ using the single-bit Byzantine broadcast primitive and send $\fe_i$ to every party. Stop updating $G_i$.
    
    \item When the above reliable broadcast delivers $\geq 2t+1$ $1$'s, perform the step $6(a)$ of Protocol Synchronous Error-free $\frac{n}{3}$-BA.
        
    \item 
    On receiving $2t+1+r$ values $maj_j$'s where $maj_j$ is sent by $P_j$, apply $\texttt{DEC}$ with $c=r$ and $d=t-r$. If  $\texttt{DEC}$ returns `failure', wait for more values.
    If $\texttt{DEC}$ returns the data $m_0, m_1, \cdots , m_t$, output $m = m_0| \cdots |m_t$.
    \end{numlist}
        
\end{mybox}

    \caption{Protocol Asynchronous Error-free $\frac{n}{3}$-RB}
    \label{fig:asyncrb:error_free}
\end{figure}

\begin{lemma}\label{lem:5}
    For any honest party $P_i$, if it obtains nonempty set $\fe_i$ in step $4$ of the protocol, then the honest parties in $\fe_i$ hold the same message of length $l$.
\end{lemma}
\begin{proof}
    Same proof of Lemma \ref{lem:1} applies.
\end{proof}

\begin{lemma}\label{lem:6}
    If the sender is honest,  then every honest party $P_i$ will eventually obtain a set $\fe_i\neq \emptyset$.
\end{lemma}

\begin{proof}
The proof of this Lemma is similar to that of Lemma \ref{lem:2}.
If the sender is honest, all honest parties eventually receive the same message $m$ in step $1$, and they will generate the same codewords. This implies all honest parties eventually will be connected to each other in $G_i$, which forms a clique of size $\geq 2t+1$. 
By the same argument from the proof of Lemma \ref{lem:2}, $\fc_i$ contains at least $t+1$ honest parties. 
Since $\fc_i$ contains at least $t+1$ honest parties and honest parties will eventually form a clique of size $\geq 2t+1$ in $G_i$, $\ff_i, \fe_i$ will subsequently contain all honest parties and have size $\geq 2t+1$.
\end{proof}

\begin{lemma}\label{lem:7}
    If some honest party outputs a message $m'$, then every honest party eventually outputs $m'$. 
\end{lemma}
\begin{proof}
    If some honest party outputs a message $m'$, then the reliable broadcast in step $6$ delivers $\geq 2t+1$ $1$'s. By the definition of reliable broadcast, eventually all honest parties will deliver an identical set of $\{0,1\}$. Therefore eventually, every honest party will deliver $\geq 2t+1$ $1$'s.

    By the same argument from the proof of Lemma \ref{lem:3},  at least $t+1$ honest parties $P_i$ send the set $\fe_i$ to all parties, and for any honest parties $P_i,P_j$ above, all honest parties in $\fe_i\cup \fe_j$ have the same message $m$.
    Then the conditions for the set $\mathbf{E'}$ in step $6$ is satisfied, which leads to identical $maj_x$. 
    Let $(s_1,...,s_n)=\texttt{ENC}(m_0,m_1,...,m_t)$ where $m=m_0|m_1|...|m_t$. We know that $maj_x=s_x$ for all $\fe_x\in \mathbf{E'}$. 
    This implies that $maj_j=s_j$ for all honest party $P_j$ in step $7$.
    On receiving $2t+1+r$ values where $0\leq r\leq t$, party $P_i$ try to decode the message with $c=r$ and $d=t-r\geq 0$ which satisfies that $2t+1+r-(t+1)\geq 2c+d$. If there are more than $r$ corrupted values, $\texttt{DEC}$ will return `failure', and wait for more values. Eventually every honest party $P_i$ will receive enough values to successfully decode $m$. Since one honest party outputs the message $m'$, we have $m=m'$.
\end{proof}

\begin{lemma}\label{lem:8}
    If the sender is honest, all honest parties eventually output the message $m$.
\end{lemma}

\begin{proof}
    If the sender is honest, by Lemma \ref{lem:6}, every honest party $P_i$ will eventually obtain its set $\fe_i\neq \emptyset$.  Then, all honest parties will broadcast a single bit of $1$, and eventually deliver at least $2t+1$ $1$'s for the broadcast. Then by the same proof of Lemma \ref{lem:7}, all honest parties will output the same message $m$.
\end{proof}

\begin{theorem}
    Protocol Asynchronous Error-free $\frac{n}{3}$-RB satisfies Termination, Agreement and Validity. 
    The protocol has communication complexity $O(nl+n^3\log n+n\fb(1))$.
\end{theorem}

\begin{proof}
    Termination is proven by Lemma \ref{lem:7} and \ref{lem:8}.
    Agreement is proved by Lemma \ref{lem:7}, and Validity is proved by Lemma \ref{lem:8}.

    Step $1$ and $2$ has communication cost $O(n l+n^2l/(t+1))=O(nl)$. Step $3$ has communication cost $O(n^3\log n)$. Step $5$ has communication cost $O(n^3+n\fb(1))$. 
    Step $6$ has communication cost $O(nl)$.
\end{proof}

\stitle{Error-free extension protocol for asynchronous Byzantine agreements under $t<n/3$.}
\label{app:errorfree:asyn_ba}
From the improved asynchronous error-free reliable broadcast protocol, we can obtain a better asynchronous error-free Byzantine agreement (ABA) protocol, directly via the same construction from \cite{patra2011error}. 
The new protocol has communication complexity $O(nl+n^4\log n+n^2\fb(1)+n\fa(1))$.

\section{A Note on Prior Results}\label{sec:previous_bug}

\begin{figure}[t!]

\begin{mybox}

    Input of every party $P_i$: An $l$-bit message $m_i$
    
    Primitive: 
    Byzantine broadcast oracle for a single bit, cryptographic collision-resistant hash function $\hash$
    
    \textbf{Checking Phase.} Every party $P_i$ does the following:
    \begin{numlist}
        \item Compute a hash of the message $m_i$ as $h_i=\hash(m_i)$ and broadcast $h_i$.
        \item Check if at least $n-t$ broadcasted hashes are equal. If no $n-t$ broadcasted hashes are equal, output $o_i=\bot$ and terminate. Otherwise, let $h$ denote the common hash value broadcasted by at least $n-t$ parties. Then form $\fp_{sm}$ as teh set of parties broadcasting $h$.
    \end{numlist}
    
    \textbf{Agreement Phase.} Every party $P_i$ does the following:
    \begin{numlist}
        \item If $P_i\in \fp_{sm}$, set output message $o_i=m_i$.
        \item Form an injective function from $\fp\setminus \fp_{sm}$ to $\fp_{sm}$, by say, mapping the party with the smallest index in $\fp\setminus \fp_{sm}$ to the party with the smallest index in $\fp_{sm}$, i.e., $\phi:\fp\setminus \fp_{sm} \rightarrow \fp_{sm}$.
        \item If $P_i\in \fp_{sm}$ and $P_i=\phi(P_j)$, then send $o_j$ to $P_j$.
        \item If $P_i\in \fp\setminus \fp_{sm}$ and received a value say, $o_j'$ from $P_j\in \fp_{sm}$ in the previous round such that $P_j=\phi(P_i)$, then check if $\hash(o_j')=h$. If the test passes, set $happy_i=1$ and assign output message $o_i=o_j'$, else set $happy_j=0$. Send $happy_i$ to all parties in $\fp$.
        \item If $P_i\in \fp\setminus \fp_{sm}$ and $happy_i=0$, then construct a set $\confset^i$ consisting of the parties $P_j, \phi(P_j)$ such that $happy_j$ received from $P_j$ in the previous step is $0$ and $P_j$ belongs to $\fp\setminus \fp_{sm}$. Set $\hmsmset^i=\fp\setminus\confset^i$, $d_i=\lceil (|\hmsmset^i|+1)/2\rceil$ and send $d_i$ to all the parties belonging to $\hmsmset^i$ and nothing to all the others.
        \item If $d_j$ is received from $P_j\in \fp\setminus \fp_{sm}$,
        \begin{numlist}
            \item Transform the message $o_i$ into a polynomial over $GF(2^c)$, for $c=\lceil l+1/d_j\rceil$ denoted by $f_i$ with degree $d_j-1$.
            \item Compute the $c$-bit piece $y_i=f_i(i)$, $H_i=(\hash(f_i(1)), \cdots, \hash(f_i(n)))$ and sends $(y_i, H_i)$ to $P_j$.
        \end{numlist}
        \item If $P_i\in \fp\setminus \fp_{sm}$ and $happy_i=0$, check each piece $y_i$ received from each $P_j\in \hmsmset^j$ against the $j$th entry of every hash value vector $H_k$ received from $P_k\in \hmsmset^i$. If at least $d_i$ of the hash values match a piece $y_i$, then accept $y_i$, otherwise reject it. Interpolate the polynomial $f$ from the $d_i$ accepted pieces $y_i$, and compute the message $m$ corresponding to the polynomial $f$. Set $o_i=m$.
        \item Output $o_i$ and terminate.
    \end{numlist}
\end{mybox}

    \caption{Protocol Synchronous Crypto. $\frac{n}{2}$-BA from \cite{ganesh2017optimal}}
    \label{fig:bug}
\end{figure}

The $(\frac{n}{2})$-BA protocol in~\cite{ganesh2017optimal} has a small flaw that the adversary can exploit to increase the communication complexity to $\Omega(n^2l)$.
In this section, we will describe the issue and provide a simple fix. 

For completeness, we provide the original protocol $(\frac{n}{2})$-BA in Figure $3$ of \cite{ganesh2017optimal}.
From the protocol, the step $6$ asks any party to send a $O(\frac{l+1}{d_j}+nk)$-bit message to some party $P_j\in \fp\setminus \fp_{sm}$, if $d_j$ is received from $P_j$.
Since $d_j=\lceil (|\fp_{\text{hmsm}}^j|+1)/2\rceil$, $|\fp_{\text{hmsm}}^j|\geq n-2t$ and $n\geq 2t+1$, it is possible that $d_j= 1$. Therefore $t$ Byzantine parties can send to all other parties the message $d_j=1$ in step $5$, which will trigger all other parties to send back messages each of length $O(\frac{l+1}{1}+nk)=O(l+nk)$ bits. Therefore, step $6$ will have communication complexity $O(tn(l+nk))=O(n^2l+n^3k)$ instead of $O(nl+n^3k)$ as claimed in \cite{ganesh2017optimal}.

Here we provide a simple fix to resolve the issue above.
Basically, we cannot allow Byzantine parties to deceive all other parties by requiring a large block of $O(l)$-bits.
Then the key step of our fix is to change the ``Send $\texttt{happy}_i$ to all parties in $\fp$'' in step $4$ to 
``Broadcast $\texttt{happy}_i$ using the single-bit broadcast oracle''.  After the broadcast, in step $5$, any honest party $P_i$ can construct identical sets $\fp_{\text{conflict}}^i=\fp_{\text{conflict}}$ and $\fp_{\text{hmsm}}^i=\fp_{\text{hmsm}}$, and compute an identical value $d_i=d=\lceil (|\fp_{\text{hmsm}}|+1)/2\rceil$.
Then in step $6$, any honest party $P_i$ in $\fp_{\text{hmsm}}$ will perform the same encoding, and send $(y_i, H_i)$ to all parties in $\fp\setminus \fp_{\text{sm}}$. Rest of the steps remain the same.

To see the communication complexity is correct after the fix, notice that only $|\fp_{\text{hmsm}}|$ honest parties will send a $O(\frac{l+1}{d}+nk)$-bit message to each party in $\fp\setminus \fp_{\text{sm}}$. Since $d=\lceil (|\fp_{\text{hmsm}}|+1)/2\rceil$, we have the communication complexity of step $6$ equals $O(|\fp\setminus \fp_{\text{sm}}|\cdot |\fp_{\text{hmsm}}|(\frac{l+1}{d}+nk))=O(nl+n^3k)$.
Also, the broadcast in step $4$ incurres $O(n\fb(1))$ cost, and rest of the protocol has the same cost, thus in total
$O(nl+n^3k+nk\fb(1))$.

\end{document}